\documentclass[final]{siamonline171218}

\usepackage{lipsum,amsfonts,graphicx,epstopdf}
\usepackage{algorithmic}
\usepackage{mathrsfs, amsmath, amscd, amssymb}
\usepackage[mathscr]{euscript} 
\usepackage{subcaption} 
\usepackage{amsopn}
\usepackage{lineno}
\usepackage{listings}
\newcommand{\tube}[1]{\boldsymbol{\dot{\mathfrak{\MakeLowercase{#1}}}}}
\newcommand{\tcol}[1]{\boldsymbol{\mathring{\mathfrak{\MakeLowercase{#1}}}}}

\newcommand{\T}[1]{\boldsymbol{\mathscr{\MakeUppercase{#1}}}} 

\newcommand{\Circ}[1]{{\rm circ}\left( #1 \right)}
\newcommand{\Fold}[1]{{\rm fold}\left( #1 \right)}
\newcommand{\Unfold}[1]{{\rm unfold}\left( #1 \right)}

\DeclareMathOperator{\bdiag} {bdiag}
\newcommand{\Ec}{ \mathbf{E}}

\DeclareMathOperator{\rank}{rank}

\DeclareMathOperator*{\argmin}{argmin}

\DeclareMathOperator{\trace}{trace}

\ifpdf
  \DeclareGraphicsExtensions{.eps,.pdf,.png,.jpg}
\else
  \DeclareGraphicsExtensions{.eps}
\fi

\usepackage{enumitem}
\setlist[enumerate]{leftmargin=.5in}
\setlist[itemize]{leftmargin=.5in}

\newsiamremark{remark}{Remark}
\newsiamremark{hypothesis}{Hypothesis}
\crefname{hypothesis}{Hypothesis}{Hypotheses}
\newsiamthm{claim}{Claim}

\headers{Fast Randomized Algorithms for Tensor Operations}{D. A. Tarzanagh and G. Michailidis}

\title{Fast Randomized Algorithms for t-Product Based Tensor Operations and Decompositions with Applications to Imaging Data\thanks{
\funding{This work was supported in part by NSF grants CCF 1540093, IIS 1632730, and NIH grant 5 R01 GM114029-03 to GM.}}}

\author{Davoud Ataee Tarzanagh\thanks{Department of Mathematics  \& UF Informatics Institute, University of Florida, Gainesville, FL, (\email{tarzanagh@ufl.edu}).} \and George Michailidis\thanks{{Department of Statistics \& UF Informatics Institute, University of Florida, Gainesville, FL, (\email{gmichail@ufl.edu}).}}}

\usepackage{amsopn}
\DeclareMathOperator{\diag}{diag}

\begin{document}

\maketitle

\begin{abstract}
Tensors of order three or higher have found applications in diverse fields, including image and signal processing, data mining, biomedical engineering and link analysis, to name a few. In many applications that involve for example time series or other ordered data, the corresponding tensor has a distinguishing orientation that exhibits a low tubal structure.
This has motivated the introduction of the tubal rank and the corresponding tubal singular value decomposition in the literature. 
In this work, we develop randomized algorithms for many common tensor operations, including tensor low-rank approximation and decomposition, together with tensor multiplication. The proposed tubal focused algorithms employ a small number of lateral and/or horizontal slices of the underlying 3-rd order tensor, that come with {\em relative error guarantees} for the quality of the obtained solutions. The performance of the proposed algorithms is illustrated on diverse imaging applications, including mass spectrometry data and image and video recovery from incomplete and noisy data. The results show both good computational speed-up vis-a-vis conventional completion algorithms and good accuracy. 
\end{abstract}



\section{Introduction}\label{sec:intro}


Tensors are multi-dimensional arrays that have been used in diverse fields of applications, including chemometrics 
\cite{smilde2005multi}, psychometrics \cite{kroonenberg1983three}, image/video and signal processing 
\cite{karahan2015tensor,kernfeld2014clustering,lutensor,semerci2014tensor,zhang2014novel} and link analysis
\cite{kolda2005higher}. They have also been the object of intense mathematical study (see for example the
review paper by Kolda and Bader \cite{kolda2009tensor} and references therein).

Analogously to the matrix case, a number of tensor decompositions have been proposed in the literature, briefly described next
for the case of 3-way tensors. Let $\T{Z} \in \mathbb{R}^{n_1\times n_2 \times n_3}$. Then, there exists a factorization, called a ``Tucker decomposition"of $\T{Z} $, of the form
\begin{eqnarray} \label{eq:Tucker}
 \T{Z}  = \sum\limits_{r_1 = 1}^{R_1} \sum\limits_{r_2 = 1}^{R_2} {
 {\sum\limits_{r_3= 1}^{R_3}{g_{r_1 r_2 r_3} \left(u^{(1)}_{r_1} \circ u^{(2)}_{r_2} \circ u^{(3)}_{r_3} \right)}}}= \T{G} \times_1 U^{(1)}\times_2 U^{(2)}  \times_3 U^{(3)},
\end{eqnarray}
where $\T{G} \in \mathbb{R}^{R_1 \times R_2 \times R_3}$ is the core tensor associated with this decomposition, $U^{(i)} \in \mathbb{R}^{n_i  \times R_i}$ is the $i$-th factor matrix and $\times_i$ is the mode$-i$ product for $i=1,2,3$. The $3$-tuple $(R_1,R_2,R_3)$ is called the Tucker rank or {\emph{multilinear rank}} of tensor $\T{Z}$ \cite{carroll1970analysis,jolliffe2002principal}.  The conventional Tucker decomposition corresponds to the orthonormal Tucker decomposition, which is also known as the higher-order SVD (HOSVD). De Lathauwer et al. \cite{de2000multilinear} proposed an algorithm to compute a HOSVD decomposition. Soon afterwards they
proposed the higher-order orthogonal iteration (HOOI) \cite{de2000best} to provide an inexact Tucker decomposition.  The CP decomposition of a tensor is another important notion of tensor-decomposition, which leads to the definition of CP rank. The CP model can be considered as a special case of the Tucker model with a superdiagonal core tensor. Further, the CP rank of a tensor equals that of its Tucker core \cite{jiang2017tensor}. 

\subsection{ Third order tensor as operator on matrices}\label{sec:tsvd:app} While the Tucker-based factorization may be sufficient for many applications, in this paper we consider an entirely different tensor decomposition based on circulant algebra  \cite{kilmer2013third}.  In this factorization, a tensor in $\mathbb{R}^{n_1 \times n_2\times n_3}$  is viewed as a
$n_1 \times n_2$ matrix of ``tubes", also known as elements of the ring $ \mathbb{R}^{n_3}$ where addition is defined as vector addition and multiplication as circular convolution. This ``matrix-of-tubes" viewpoint leads to definitions of a new multiplication for tensors (``tubal multiplication"), a new rank for tensors (``tubal rank"), and a new notion of a 
singular value decomposition (SVD) for tensors (``tubal SVD"). The tubal SVD (t-SVD) provides the ``best" 
tubal rank-$r$ approximation to $\T{Z}$, as measured with respect to any unitary invariant tensor norm. 

A limitation of the t-SVD decomposition is that it depends directly on the orientation of the tensor, whereas the CP and Tucker decompositions are not. This suggests that the latter decompositions are well suited for data applications where the tensor's orientation is not critical - e.g. chemometrics \cite{smilde2005multi} and/or psychometrics \cite{kroonenberg1983three}. However, in applications involving time series or other ordered data, the orientation of the tensor is fixed. Examples include, but not limited to,  computed tomography (CT)  \cite{semerci2014tensor}, facial recognition \cite{hao2013facial} and video compression \cite{zhang2014novel}, where the tensor decomposition is dependent on the third dimension. Analogous to compressing a two-dimensional image using the matrix SVD (a classic linear algebra example, with detailed writeup in \cite{kalman1996singularly}), the t-SVD decomposition can be used to compress several images taken over time (e.g. successive frames from a video). Since such images do not change substantially from frame to frame, we expect tubal compression strategies to provide better results than performing a matrix SVD on each separate image frame. The former consider the tensor as a whole object, rather than as a collection of distinct images \cite{kilmer2011factorization,kilmer2013third,zhang2014novel,semerci2014tensor,lutensor}. Further, t-SVD is essentially based on a group theoretic approach, where the multidimensional structure is unraveled by constructing group-rings along the tensor fibers \footnote{We consider the group rings constructed out of cyclic groups, resulting in an algebra of circulants. However, the results presented in this paper hold true for the general group-ring construction.}. The advantage of such an approach over existing ones is that the resulting algebra and corresponding
analysis enjoys many similar properties to matrix algebra and analysis. For example, it is shown in \cite{huang2014provable} that recovering a 3-way tensor of length $n$ and Tucker rank $(r, r, r)$
from random measurements requires $O(rn^2 \log^2(n))$ observations under a matrix coherence condition on all mode-$n$ unfoldings. However, the number of samples needed for exact recovery of a 3-way tensor of length $n$ and tensor tubal-rank $r$ is $O(rn^2 \log(n^2))$ under a weaker tensor coherence assumption \cite{zhang2017exact}. Further, consider the decomposition $\T{X} = \T{L} + \T{E}$, where $\T{L} \in \mathbb{R}^{n_1 \times n_2 \times n_3}$ is low rank and $\T{E}$ is sparse. Let $n_{(1)} =\max(n_1, n_2)$ and $n_{(2)} =\min(n_1, n_2)$. The work in \cite{lutensor} shows that for tensor $\T{L}$ with coherence parameter $\mu$, the recovery is guaranteed with high probability for the tubal rank of order $n_{(2)} n_3/(\mu(\log n_{(1)} n_3)^2)$ and a number of nonzero entries in $\T{E}$ of order $O(n_1 n_2 n_3)$. Hence, under the same coherence condition (see, Definitions~\ref{eq:ten:incoh} and \ref{eq:ten:incoh1}), the tubal robust tensor factorization problem perfectly recovers the low-rank and sparse components of the underlying tensor.


A shortcoming of these three classical decompositions is their brittleness with respect to severely corrupted or outlier data entries. To that end, a number of approaches have been developed in the literature to recover a low-rank tensor representation from data subject to noise and corrupted entries. We focus on two instances of the problem based on the t-SVD algorithm: (i) noisy tensor completion, i.e., recovering a low-rank tensor from a small subset of noisy entries, and (ii) noisy robust tensor factorization, i.e., recovering a low-rank tensor from corrupted data by noise and/or outliers of arbitrary magnitude \cite{lutensor,zhang2014novel}. These two classes of tensor factorization problems have attracted significant interest in the research community \cite{bouwmans2015decomposition,candes2011robust,semerci2014tensor,zhang2014novel}.
In particular, convex formulations of noisy tensor factorization have been shown to exhibit strong theoretical recovery guarantees and a number of algorithms has been developed for solving them \cite{zhang2014novel,semerci2014tensor,lutensor}. 

It is frequently mentioned that (noisy) tensor factorization, despite its numerous advantages, also exhibits a number of drawbacks listed below:
\begin{itemize}
  \item The available methods \cite{braman2010third,kilmer2011factorization,kilmer2013third,zhang2014novel,semerci2014tensor,lutensor} are inherently sequential and all rely on the repeated and costly computation of t-SVD factors, Discrete Fourier Transform (DFT) and its inverse (IDFT), that limit their scalability.
  \item The basis tensor vectors resulting from t-SVD have little concrete meaning, which makes it challenging for practitioners to interpret the obtained results. For instance, the vector [$(1/2)$ age - $(1/\sqrt{2})$ height + $(1/2)$income], being one of the significant uncorrelated factors from a data set of people's features is not easily interpretable (see discussion in \cite{drineas2008relative}). Kuruvilla et al. in \cite{kuruvilla2002vector}
have also claimed: ``it would be interesting to try to find basis vectors for all experiment vectors, using actual experiment vectors and not artificial bases that offer little insight".
  \item The t-SVD decomposition for sparse tensors does not preserve sparsity in general, which for large size tensors leads to excessive	computations and storage requirements. Hence, it is important to compute low-rank tensor factorizations that preserve such structural properties of the original data tensor.
\end{itemize}

\subsection{Main Contributions}\label{sec:contr} In this paper, we study scalable \emph{randomized} tensor multiplication 
(rt-product algorithm) and tensor factorization (rt-project) operations and extend the matrix CX and CUR type decompositions \cite{deshpande2006matrix, drineas2006fast, drineas2008relative,mackey2011divide} to third order tensors using a circulant algebra embedding. To that end, we develop a basic algorithm (t-CX), together with a more general one (t-CUR) based on tensor slice selection that come with relative error guarantees. Finally, we propose a new tensor nuclear norm minimization method, called CUR t-NN, which solves the noisy tensor factorization problem using a small number of lateral and/or horizontal slices of the underlying tensor. Specifically, CUR t-NN uses an adaptive technique to sample slices of the tensor based on a \emph{leverage score} for them and subsequently solves a convex optimization problem followed by a projection step to recover a low rank approximation to the underlying tensor. Advantages of CUR t-NN include:
\begin{itemize}
  \item Contrary to nuclear norm minimization based approaches,
which minimize the sum of the singular values of the underling tensor, our approach only minimizes the sum of singular values of the set of sampled slices corresponding to the largest leverage scores. Thus, we obtain a more accurate and robust approximation to the rank function.
  \item Using subspace sampling, we obtain provable relative-error recovery guarantees for tubal product based tensor factorization. This technique is likely to be useful for other tensor approximation and data analysis problems.
  \item The proposed algorithm for noisy tensor factorization tasks has polynomial time complexity.
\end{itemize}

\subsection{Related work on randomized tensor factorization}\label{sec:related}
Note that there has been prior work on the topic of randomized methods for tensor low rank approximations/decompositions. For example, recent work includes \cite{mahoney2008tensor, caiafa2010generalizing}. Our methods are different from these studies, since we rely on the \emph{t-product} construct in \cite{braman2010third,kilmer2013third,kilmer2011factorization} to provide \textit{tubal} tensor multiplication and low rank approximation. In \cite{mahoney2008tensor}, the proposed tensor-CUR algorithm employs a CUR matrix approximation to one of the unfolding matrix modes (the distinguished mode) providing an approximation based on few tube fibers and few slices.  Hence, that algorithm achieves an additive error guarantee in terms of the flattened (unfolded) tensor, rather than the original one. However, the algorithms presented in this work offer \textit{relative error guarantees} for the obtained approximations, in terms of the original low tubal rank tensor. Further, we propose a new tensor nuclear norm minimization method, which solves the noisy tensor factorization problem in the fully and partially observed setting using a small number of lateral and/or horizontal slices of the underlying tensor. It is worth mentioning that randomized algorithms were used to efficiently solve the robust matrix PCA problem using small sketches constructed from random linear measurements of the low rank matrix \cite{rahmani2017randomized, zhou2011godec,liu2011solving,mackey2011divide,mu2011accelerated}.  Our proposed nuclear norm minimization approach is different from these studies, since we relay on slice selection and projection (mainly t-CUR factorization). Further, our proposed algorithm uses an
adaptive sampling strategy and provides high probability recovery guarantees of the exact low rank tensor approximation under a weaker coherence assumption.
 
The remainder of the paper is organized as follows. In Section~2, we review some relevant mathematical concepts including the tensor circulant algebra, basic definitions and theorems of tensors, and the t-SVD decomposition based on the t-product concept. In Section~3, we provide the randomized tensor decompositions and provide their relative error guarantees and introduce the concept of slice selection and projection. In Section~4, we introduce, evaluate and analyze our proposed algorithm (CUR t-NN) for large scale noisy tensor decomposition. Experimental results are presented in Section~5 and some concluding remarks are drawn in Section~6.

The detailed proofs of the main results established are delegated to the Appendix.

\section{Mathematical preliminaries}\label{sec:notation}

Next, we introduce key definitions and concepts used in subsequent developments.

\begin{description}
\item[Tensor indexing.] We denote tensors by boldface Euler script letters, e.g., $\mathscr{X}$, matrices by boldface capital letters, e.g., $X$, vectors by lowercase letters, e.g., $x$. The order of a tensor is the number of dimensions (also refereed to as ways or modes). In this work, we focus on 3-way tensors.
\item[Fibers and slices \cite{Zhang2016randtensor}.] A {\em fiber} of tensor $\mathscr{X}$ is a one-dimensional array defined by fixing two of the indices. 
Specifically, $\mathscr{X}_{:jk}$ is the (j,k)-th {\em column} fiber, $\mathscr{X}_{i:k}$ is the (i,k)-th {\em row} fiber, and $\mathscr{X}_{ij:}$ is the (i,j)-th {\em tube} fiber. A slice of tensor $\mathscr{X}$ is a two-dimensional array defined by fixing one index only. Specifically, $\mathscr{X}_{i::}$ is the i-th {\em horizontal} slice, $\mathscr{X}_{:j:}$ is the j-th {\em lateral} slice, and $\mathscr{X}_{::k}$ is the k-th {\em frontal} slice. For convenience, $\mathscr{X}_{::k}
$ is written as $\mathscr{X}_{k}$. The vectorization of $\mathscr{X}$ is denoted by $vec(\mathscr{X})$. For a 3-way tensor $\mathscr{X} \in \mathbb{R}^{n_1 \times n_2 \times n_3}$, we denote its $(i,j,k)$-th entry as $\mathscr{X}_{ijk}$.
 \begin{figure}[!h]
 \centering
 \begin{subfigure} {0.17\textwidth}
    {\includegraphics[width=\textwidth]{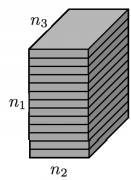}}
    \caption{\scriptsize{horizontal slices}}
  \end{subfigure}
  \begin{subfigure} {0.17\textwidth}
  {\includegraphics[width=\textwidth] {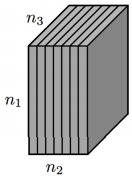}}
    \caption{\scriptsize{lateral slices}}
  \end{subfigure}
  \begin{subfigure} {0.17\textwidth}
  {\includegraphics[width=\textwidth]{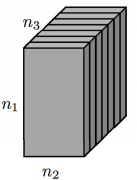}}
    \caption{\scriptsize{frontal slices}}
  \end{subfigure}
\caption{Slices of an $n_1 \times n_2 \times n_3$ tensor $\mathscr{X}$.}
\end{figure}
\item[Norms.] We denote the $\ell_1$ norm as $\|\T{X}\|_1 := \sum_{ijk}|x_{ijk}|$, the infinity norm as $\|\T{X}\|_{\infty} := \max_{ijk}|x_{ijk}|$, and the Frobenius norm as  $\|\T{X}\|_F := \sqrt{\sum_{ijk}|x_{ijk}|^2}$. The above norms reduce to the vector or matrix norms if $\T{X}$ is a vector or a matrix.
\item[Operators.] For $\T{X} \in \mathbb{R}^{n_1\times n_2 \times n_3}$, and using the Matlab commands \texttt{fft,ifft}, we denote by $\hat{\T{X}}$ the result of applying the Discrete Fourier Transform (DFT) on $\T{X}$ along the 3-rd dimension, i.e., $\hat{\T{X}} = \texttt{fft}(\T{X},[],3)$. Analogously, one can also compute $\T{X}$ from $ \hat{\T{X}}$ via the Inverse Discrete Fourier Transform (IDFT), using  $\texttt{ifft}(\hat{\T{X}},[],3)$. In particular, we denote by $\hat{X}$ the block diagonal matrix with each blockcorresponding to the frontal slice $\hat{X}_{::k}$ of $ \hat{\T{X}}$.
i.e.,
\begin{equation*}
\hat{X} = \bdiag_{k \in [n_3]} (\hat{\T{X}}_{::k}):= \begin{pmatrix}
       \hat{\T{X}}_{::1}
       \\&\hat{\T{X}}_{::2}\\\
        &&\ddots\\
        &&&\hat{\T{X}}_{::n_3}\\\
    \end{pmatrix}.
\end{equation*}
\end{description}

\subsection{Tensor basics}\label{sec:tensor-basics} Next, we review relevant mathematical concepts including the tensor SVD (t-SVD), basic definitions and operations and other technical results of tensors \cite{kilmer2011factorization,semerci2014tensor,zhang2017exact}, that are used throughout the paper.

\begin{definition}\label{circ}\textbf{(Tensor product)}.
Given two tensors $\T{Z} \in \mathbb{R}^{n_1\times n_2 \times n_3}$ and $ \T{X} \in \mathbb{R}^{n_2 \times n_4 \times n_5}$, the t-product $\T{Z} * \T{X}$ is the $n_1 \times n_4 \times n_3$ tensor,
\begin{equation}\label{tprod}
    \T{C} = \T{Z}* \T{X} = \Fold{\Circ{\T{Z}}}. \Unfold{\T{X}},
\end{equation}
where
$$
\Circ{\T{Z}}:=
\begin{pmatrix}
    \T{Z}_{::1}&  \T{Z}_{::n_3}&  \dots  & \T{Z}_{::2} \\
    \T{Z}_{::2}& \T{Z}_{::1}& \dots  & \T{Z}_{::3} \\
    \vdots&  \vdots &\ddots & \vdots \\
   \T{Z}_{::n_3}& \T{Z}_{::n_3-1}& \dots  & \T{Z}_{::1}
\end{pmatrix},
$$
and
$$
\Unfold{\T{X}}:=
\begin{pmatrix}
    X_{::1} \\
    X_{::2} \\
    \vdots \\
    X_{::n_3}
\end{pmatrix},
\qquad  \qquad \Fold{\Unfold{\T{X}}}=\T{X}.
  $$
\end{definition}

Because the circular convolution of two tube fibers can be computed by the DFT, the t-product can be alternatively computed in the Fourier domain, as shown in Algorithm~\ref{algtt}.

\begin{algorithm}
\caption{t-product $\T{C} = \T{Z}* \T{X}$ in the Fourier domain}
\label{algtt}
\begin{algorithmic}[1]
\STATE{\textbf{Input}: $\T{Z} \in \mathbb{R}^{n_1 \times n_2 \times n_3}$; $\T{X} \in \mathbb{R}^{n_2 \times n_4 \times n_3}$}
\STATE{$\hat{\T{Z}} \leftarrow \texttt{fft}(\T{Z}, [], 3);$}
\STATE{$\hat{\T{X}} \leftarrow \texttt{fft}(\T{X}, [], 3);$}
\FOR {$k=1, \dots, n_3$}
\STATE{$ \hat{\T{C}}_{::k}= \hat{\T{Z}}_{::k}\hat{\T{X}}_{::k};$}
\ENDFOR
\STATE{$\T{C} \leftarrow \texttt{ifft}(\hat{\T{C}}, [], 3);$}
\RETURN $\T{C}$
\end{algorithmic}
\end{algorithm}

\begin{definition}\textbf{(Conjugate transpose)}.
 The conjugate transpose of a tensor $\T{X} \in \mathbb{R}^{n_1\times n_2 \times n_3}$ is tensor $\T{X}^* \in \mathbb{R}^{n_2\times n_1 \times n_3}$ obtained by conjugate transposing each of the frontal slices and then reversing the order of transposed frontal slices 2 through $n_3$.
\end{definition}

\begin{definition}\textbf{(Identity tensor)}.
The identity tensor $\T{I} \in \mathbb{R}^{n\times n \times n_3}$ is the tensor whose first frontal slice is the $n \times n$ identity matrix, and whose other frontal slices are all zeros.
\end{definition}

\begin{definition}\label{orth}\textbf{(Orthogonal tensor).}
  A tensor $\T{Q} \in \mathbb{R}^{n\times n \times n_3}$ is orthogonal if it satisfies $\T{Q}^* * \T{Q} =\T{Q} * \T{Q}^* = \T{I}.$
\end{definition}

\begin{definition}\label{diag}
 \textbf{(F-diagonal Tensor)}. A tensor is called $f-$diagonal if each of its frontal slices is a diagonal matrix.
\end{definition}

\begin{theorem}\textbf{(t-SVD)}.
  Let $\T{X} \in \mathbb{R}^{n_1\times n_2 \times n_3}$. Then, it can be factored as
  \begin{equation}\label{tsvd}
    \T{X}= \T{U} * \Sigma *\T{V}^T,
  \end{equation}
 where  $\T{U} \in \mathbb{R}^{n_1\times n_2 \times n_3}$, $\T{V} \in \mathbb{R}^{n_1\times n_2 \times n_3}$ are orthogonal and $\Sigma \in \mathbb{R}^{n_1\times n_2 \times n_3}$ is a f-diagonal tensor.
\end{theorem}

Note that t-SVD can be efficiently computed based on the matrix SVD in the Fourier domain. This is based on the key property that the block circulant matrix can be mapped to a block diagonal matrix in the Fourier domain, i.e.
\begin{equation}\label{circ1}
  (F_{n_3}\otimes I_{n_1})\cdot \Circ{\T{X}}\cdot (F^{-1}_{n_3} \otimes I_{n_2})= \hat{X},
\end{equation}
where $F_{n_3}$ denotes the $n_3 \times n_3$ DFT matrix and $\otimes$ denotes the Kronecker product.

\begin{definition}(\textbf{Tensor multi and tubal rank}).
 The tensor multi-rank of $\T{X} \in \mathbb{R}^{n_1 \times n_2 \times n_3}$ is a vector $\upsilon \in \mathbb{R}^{n_3}$ with its $i-th$ entry being the rank of the i-th frontal slice of $\hat{\T{X}}$, $r_i = \rank(\hat{\T{X}}_{i::})$. The tensor tubal rank, denoted by $r =\rank_t(\T{X})$, is defined as the number of nonzero singular tubes of $\Sigma$, where $\Sigma$ is obtained from the t-SVD of $\T{X} = \T{U} \ast \Sigma \ast \T{V}^T$.
That is,
\begin{equation}\label{tubl}
 r = \text{card}\{i: \Sigma_{ii:} \neq 0\} = \max_{i} r_i,
\end{equation}
where card denotes the cardinality of a set.
\end{definition}
The tensor tubal rank shares some properties
of the matrix rank, e.g.
$$ \rank_t(\T{X} \ast \T{Z}) \leq \min (\rank_t(\T{X}), \rank_t(\T{Z})).$$

Many tensor completion and decomposition techniques for video and seismic noise reduction rely on a low-rank factorization of a time-frequency transform. Further, certain energy methods broadly used in image processing, e.g., PDEs \cite{bertalmio2000image} and belief propagation techniques mainly focus on local relationships. The basic assumption is that the missing entries depend primarily on their neighbors. Hence, the further apart two pixels are, the smaller their dependance is. However, for video and time series of images the value of the missing entry also depends on  entries which are relatively far away in the time/sequence dimension. Thus, it is necessary to develop a tool to directly
capture such global information in the data.  Using \eqref{tubl},  $r_0$ is the rank of the ``mean image" across the video sequence. Meanwhile, $r_1$ is the rank of the next frequency's content across frames, etc. Under a smoothness assumption  
that captures global information at given pixels across time, the frontal slices (after FFT) for bigger $i$ have smaller singular values.

\begin{definition}(\textbf{Tensor nuclear norm}). The tensor nuclear norm of a tensor $\T{X} \in \mathbb{R}^{n_1 \times n_2 \times n_3}$, denoted by $\|\T{X}\|_{\circledast}$, is defined as the average of the nuclear norm of all frontal slices of $\hat{\T{X}}$, i.e.
\begin{equation}\label{tnuc}
\|\T{X}\|_{\circledast}:= \frac{1}{n_3}\sum_{k=1}^{n_3}\| \hat{\T{X}}_{::k}\|_{*}.
\end{equation}
\end{definition}

 The above tensor nuclear norm is defined in the Fourier domain. It is closely related to the nuclear norm of the block circulant matrix in the original domain. Indeed,
 \begin{eqnarray*}
 \nonumber
   \|\T{X}\|_{\circledast} &=& \frac{1}{n_3}  \sum_{k=1}^{n_3} \|\hat{\T{X}}_{::k}\|_{*} = \frac{1}{n_3}  \| \hat{X} \|_{*} \\
   \nonumber
    &=& \frac{1}{n_3} \|(F_{n_3} \otimes I_{n_1}) \cdot \Circ{\T{X}} \cdot (F^{-1}_{n_3} \otimes I_{n_2})\|_{*}\\
    &=& \frac{1}{n_3} \|\Circ{\T{X}}\|_{*}.
 \end{eqnarray*}

 The above relationship gives an equivalent definition of the tensor nuclear norm in the original domain. Thus, the tensor nuclear norm is the nuclear norm (with a factor $1/n_3$) of a new matricization (block circulant matrix) of a tensor.

\begin{definition}\label{def:spec:ten}
	\textbf{(Tensor spectral norm)}.  The tensor spectral norm of $\T{X} \in \mathbb{R}^{n \times n \times n_3}$, denoted as $ \| \T{X} \|$, is defined as $$\|\T{X}\|:= \|\hat{X}\|_2 = \|(F_{n_3} \otimes I_{n_1}) \cdot \Circ{\T{X}} \cdot (F^{-1}_{n_3} \otimes I_{n_2})\|_{2},$$ where $\|\cdot\|_2$ denotes the spectral norm of a matrix.
\end{definition}

\begin{definition}\textbf{(Inverse of tensor)}.
The inverse of $\T{X} \in \mathbb{R}^{n \times n \times n_3}$, denoted by $\T{X}^{-1}$ satisfies
\begin{equation}
\T{X}^{-1} * \T{X} = \T{X} * \T{X}^{-1} = \T{I},
\end{equation}
where $\T{I}$ is the identity tensor of size $n \times n \times n_3$.
\end{definition}

\begin{definition}\label{def:basis}\textbf{(Standard tensor basis).}
The \textbf{lateral basis} $\tcol{e}_i$, is of size $n_1 \times 1 \times n_3$ with only one entry equal to 1 and the remaining equal to zero, in which the nonzero entry 1 will only appear at the first frontal slice of $\tcol{e}_i$. Normally its transpose $\tcol{e}_i^\top$ is called the \textbf{horizontal basis}. The other standard tensor basis is called \textbf{tube basis} $\tube{e}_i$, and corresponds to a tensor of size $1 \times 1 \times n_3$ with one entry equal to 1 and the rest equal to 0. Figure~\ref{fig:basis} illustrates these bases.
\end{definition}
\begin{figure}[h]
\centering \makebox[0in]{
    \begin{tabular}{c c}
      \includegraphics[scale=0.25]{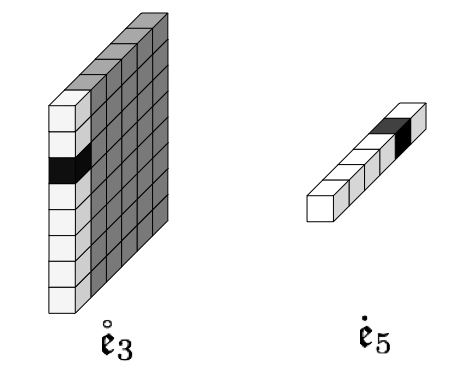}
 \end{tabular}}
  \caption{ The lateral basis $\tcol{e}_3$ and the tube basis $\tube{e}_5$. The black cubes are 1, gray and white cubes are 0. The white cubes stand for the potential entries that could be 1.}
  \label{fig:basis}
\end{figure}

\subsection{Linear algebra with tensors: free submodules}
\label{subsec:free}

The set of complex numbers $\mathbb{C}$ with standard scalar addition and multiplication is a field and $\mathbb{C}^{n_3}$ forms a vector space over this field. However, as pointed out in \cite{kernfeld2014clustering, kilmer2011factorization}, the set of tubes $\mathbb{C}^{1\times 1 \times n_3}$ equipped with the tensor product form a ring with unity.  A module over a ring can be thought of as a generalization of a vector space over a field, where the corresponding scalars are the elements of the ring.
In linear algebra over a ring, the analog of a subspace is a {\it free submodule}. Our algorithm relies on
submodules, and the following theorem.

\begin{theorem} \label{thm:free}
The set of slices
\begin{equation}\label{set:slice}
\Upsilon := \{\T{X}_{:j:} \mid \quad \T{X}_{:j:} \in \mathbb{C}^{n_1 \times 1 \times n_3}, \quad j \in [n_2]\},
\end{equation}
forms a free module over the set of tubes $\mathbb{C}^{1 \times 1 \times n_3}$.
\end{theorem}
\begin{proof}
A detailed proof is provided in \cite{braman2010third}.
\end{proof}

Using Theorem~\ref{thm:free}, $\Upsilon$ has a \emph{basis} so that any element of it can be written as a ``t-linear combination" of elements of a set of basis slices. A ``t-linear combination", is defined as a sum of slices multiplied, based on the t-product, by coefficients from $\mathbb{C}^{1 \times 1 \times n_3}$.


\begin{definition} (\textbf{Slice-wise linear independence})\label{def:lin_ind}
The slices in a subset $\Lambda=\{\T{X}_{:1:}, \dots, \T{X}_{:n:}\}$ of $\Upsilon$ are said to be \emph{linearly dependent}, if there exists a finite number of distinct slices $\T{X}_{:1:}, \dots, \T{X}_{:m:}$ in $\Lambda$, and tubes $ \T{C}_{11:}, \dots, \T{C}_{mm:}$, not all zero, such that
\begin{equation}\label{lin:indepndence}
\sum_{j=1}^{m} \T{X}_{:j:} \ast \T{C}_{jj:}= O,
\end{equation}
where $O$ denotes the lateral slice comprising of all zeros.

The slices in a subset $\Lambda=\{\T{X}_{:1:}, \dots, \T{X}_{:n:}\}$ of $\Upsilon$ are said to be \emph{linearly independent} if the equation
$$
\sum_{j=1}^{n} \T{X}_{:j:} \ast \T{C}_{jj:}= O,
$$
can only be satisfied by all zero tubes $\T{C}_{jj:},~ j=1, \dots, n$.
\end{definition}

Based on the computable t-SVD, the tensor nuclear norm \cite{semerci2014tensor} is used to replace the tubal rank for low-rank tensor recovery (from incomplete/corrupted tensors) by solving the following convex program,
\begin{equation}\label{tlow6}
  \min_{\T{L}} \|\T{L}\|_{\circledast} \qquad s.t.  \qquad
  \| P_{\Omega} (\T{X} - \T{L})\| \leq \Delta,
\end{equation}
where $\|\T{L}\|_{\circledast}$  denotes the tensor nuclear norm and
\begin{eqnarray*}
  (P_{\Omega} (\T{X})_{ij} &=& \T{X}_{ij}, \quad \text{if}  \quad (i,j) \in \Omega  \quad \text{and} \quad (P_{\Omega}(\T{X}))_{ij}= 0  \quad \text{otherwise}.
\end{eqnarray*}

Similarly to the matrix completion problem, recovery of tensor $\T{X}$ from its observed entries is essentially infeasible if the large majority of the entries are equal to zero \cite{candes2009exact}. For the tensor completion case, it is the case that if tensor $\T{X}$ only has a few entries which are not equal to zero, in its t-SVD  $\T{U}*\T{S}*\T{V}^\top = \T{X}$, the singular tensors $\T{U}$ and $\T{V}$ will be highly concentrated. Indeed, not all tensors can be recovered from data sets with missing entries and/or large outliers. Hence, in analogy to the main idea in matrix completion \cite{candes2009exact}, tensor slices $\T{U}(:,i,:)$ and $\T{V}(:,i,:), i=1,2,...,r$, need to be sufficiently spread out, which in turn implies that they should be uncorrelated with the standard tensor basis.
Our analysis in Section~\ref{noisy:completion} will focus on noisy tensor completion based on a robust factorization algorithm whose estimation/recovery guarantees are expressed in terms of the coherence of the target low-rank tensor $\T{L}$. It establishes that lower values of tensor coherence provide better recovery results. In addition, we propose a randomized approximation algorithm, whose guarantees are also related to the notion of tensor coherence. The following three notions of tensor coherence are defined next.

\begin{definition}\label{eq:ten:incoh}(\textbf{Tensor $\mu_0$-coherence}).
Let $\T{V} \in \mathbb{R}^{n \times r \times n_3}$ contain orthonormal lateral basis with $r \leq n$. Then, the $\mu_0$-coherence of $\T{V}$ is given by:
\begin{eqnarray*}
  \mu_0(\T{V}) &:=& \frac{n n_3}{r} \max_{1\leq i \leq n} \|\T{V}^\top * \tcol{e}_i\|_F^2=\frac{n n_3}{r} \max_{1\leq i\leq n} \|\T{V}_{i::}\|_F^2,
\end{eqnarray*}
where $\tcol{e}_i$ is standard lateral basis.
\end{definition}

\begin{definition}\label{eq:ten:incoh1}(\textbf{Tensor $\mu_1$-coherence}). For a tensor $\T{L} \in \mathbb{R}^{n_1 \times n_2 \times n_3}$ assume that $\text{rank}_t(\T{L}) = r$. Then, the $\mu_1$-coherence of $\T{L}$ is defined as
\begin{eqnarray*}
  \mu_1(\T{L}) := \frac{n_1 n_2 n_3^2}{r}\| \T{U}* \T{V}^T\|^2_{\infty},
\end{eqnarray*}
where $\T{L}$ has the skinny t-SVD $\T{L}= \T{U}* \T{S}*\T{V}^\top $.
\end{definition}
\begin{definition}\label{eq:ten:incohh01}(\textbf{Tensor $(\mu,r)$-coherence}).
For any $\T{\mu} >0$, we call a tensor $\T{L}$ $(\mu,r)$-coherent if
\begin{eqnarray*}
\text{rank}_t(\T{L}) &=& r, \\
  \max\{\mu_0(\T{U}), \mu_0(\T{V})\} &\leq& \mu, \\
  \mu_1(\T{\T{L}}) &\leq& \mu.
\end{eqnarray*}
\end{definition}
Note that the standard tensor coherence condition is much weaker than the \emph{matrix weak coherence} one for each frontal slice of $\hat{\T{X}}$ \cite{zhang2017exact}. Hence, in the analysis of the proposed randomized algorithms, we will use the standard tensor coherence condition.

\section{Approximate Tensor Low Rank Decompositions}

We develop extensions of the matrix CX and CUR decompositions \cite{drineas2006fast,drineas2008relative} to third-order tensors. The proposed algorithms result in computing low-rank tensor approximations that are explicitly expressed in terms of a small number of slices of the input tensor. We start by introducing a basic computational tool -a randomized tensor multiplication procedure- that is used in subsequent developments.

\subsection{Approximate tensor multiplication}\label{sxn:matrix_multiply:alg}

Next, we present a randomized tensor-product and provide key results on the quality of the resulting approximation. Given two tensors $\T{A}$ and $\T{B}$, using Definition~\ref{circ}, the t-product may be written as follows:
\begin{equation}\label{prods}
  \T{A}*\T{B} = \sum_{t=1}^{n_2}\T{A}_{:t:}*\T{B}_{t::},
\end{equation}
where $*$ denotes the tensor product.

It can be easily shown that the left hand side of \eqref{prods} is equivalent to the block multiplication and then summation in the Fourier domain. When tensor multiplication is formulated as \eqref{prods}, we can develop a randomized algorithm to approximate the product $\T{A}*\T{B}$.

Algorithm~\ref{alg:matrix multiply} takes as input two tensors $\T{A} \in \mathbb{R}^{n_1 \times n_2 \times n_3}$ and $\T{B} \in \mathbb{R}^{n_2 \times n_4 \times n_3}$, a positive integer $ c \le n_2$, and a probability distribution $\{p_i\}_{i=1}^{n_2}$ over $[n_2]$. It returns as output two tensors $\T{C}$ and $\T{R}$, where the lateral slices of $\T{C}$ correspond to a small number of sampled and rescaled slices of $\T{A}$, and similarly the horizontal slices of $\T{R}$ constitute a small number of sampled and rescaled slices of $\T{B}$. Specifically, consider at most $c$ lateral slices (in expectation) of $\T{A}$ selected, with the $i$-th lateral slice of $\T{A}$ in $\T{C}$ be chosen with probability $p_i = \min \{1,cp_i\}$. Then, define the sampling tensor $\T{S} \in \mathbb{R}^{n_2 \times n_2 \times n_3}$ to be binary where $S_{ii1} = 1$ if the $i$-th slice is selected and $\T{S}_{ij2:n_3} = 0$ otherwise. Define the rescaling tensor $\T{D} \in \mathbb{R}^{n_2 \times c \times n_3}$ to be the tensor with $\T{D}_{ij1} = 1/\sqrt{cp_{j}}$, if $i-1$ of the previous slices have been selected and $\T{D}_{ij2:n_3}=0$ otherwise. Analogously, the {rt-product} samples and rescales the corresponding horizontal slices of tensor $\T{B}$. In both cases, $\T{C} = \T{A}* \T{S}*\T{D}$ is an $n_1 \times c \times n_3$ tensor consisting of sampled and rescaled copies of the lateral slices of $\T{A}$, and $\T{R} = (\T{S}*\T{D})^\top *\T{B} = \T{D}*\T{S}^\top *\T{B}$ is a $c \times n_4 \times n_3$ tensor comprising of sampled and rescaled copies of the horizontal slices of $\T{B}$. For the case of $n_3=1$, the algorithm selects column-row pairs in Algorithm~\ref{alg:matrix multiply} and is identical to the algorithm in \cite{drineas2006fast}.

\begin{algorithm}[h]
\caption{\texttt{(rt-product)}, a fast Monte-Carlo algorithm for approximate tensor multiplication.}\label{alg:matrix multiply}
\begin{algorithmic}[1]
\STATE{ \textbf{Input}: $\T{A} \in \mathbb{R}^{n_1 \times n_2 \times n_3}$, $\T{B} \in \mathbb{R}^{n_2 \times n_4 \times n_3}$, $p_i \geq 0, i\in[n_2]$ s.t. $\sum_{i \in [n_2]}p_i=1$, positive integer $c \leq n_2$.}
\STATE{Initialize $\T{S} \in \mathbb{R}^{n_2 \times n_2 \times n_3} $ and $\T{D} \in \mathbb{R}^{n_2 \times c \times n_3}$ to the all zeros tensors;}
\STATE{$t=1$;}
\FOR{$i=1, \dots, n_2$}
\STATE{Pick $i$ with probability $\min\{1,c p_i\}$;}
\IF{$i$ is picked}
\STATE{$\T{S}_{it1} = 1, \quad \T{S}_{it2:n_3} =0$;}
\STATE{$\T{D}_{tt1} = 1/\min\{1,\sqrt{cp_{i}}\}, \quad \T{D}_{tt2:n_3} =0;$}
\STATE{$t = t + 1;$}
\ENDIF
\ENDFOR
\STATE{$\T{C} = \T{A}*\T{S}*\T{D}$,  \quad $\T{R} = \T{D}* \T{S}^\top * \T{B} $;}
\STATE{$\hat{\T{C}} \leftarrow \texttt{fft}(\T{C}, [], 3)$, \quad  $\hat{\T{R}} \leftarrow \texttt{fft}(\T{R}, [], 3)$;}
\FOR{$k=1, \dots, n_3$}
\STATE{$ \hat{\T{Z}}_{::k}= \hat{\T{C}}_{::k}\hat{\T{R}}_{::k};$}
\ENDFOR\\
\STATE{$\T{Z} \leftarrow \texttt{ifft}(\hat{\T{Z}}, [], 3);$}\\
\RETURN  $\T{Z}$
\end{algorithmic}
\end{algorithm}
\subsection{Running time of rt-product}

 The {rt-product} is computationally efficient, has small memory requirements, and is well suited for large scale problems. Indeed, the {rt-product} algorithm can be implemented without storing tensors $\T{A} \in \mathbb{R}^{n_1 \times n_2 \times n_3}$ and $\T{B} \in \mathbb{R}^{n_2 \times n_4 \times n_3}$ and uses $O(\max (n_1,n_4) c n_3 \log(n_3))$ flops to transform 
the input to the Fourier domain and $O( c (n_1 +n_2 + n_4) n_3)$ flops to construct $\T{C}\in \mathbb{R}^{n_1\times c \times n_3}$ and $\T{R}\in \mathbb{R}^{c \times n_4 \times n_3}$, where \begin{equation}\label{multi:approx}
\T{A}*\T{B}\sim \T{C}*\T{R}.
\end{equation}
It is worth mentioning that we do not require storing the sampling and rescaling tensors $\T{S}$ and $\T{D}$ in our implementation. 

Next, we provide the main result on the quality of the approximation obtained by Algorithm~\ref{alg:matrix multiply} that
specifies the assumptions under which \eqref{multi:approx} holds. The most interesting of these assumptions is that the sampling probabilities used to randomly sample the lateral slices of $\T{A}$ and the corresponding horizontal slices of $\T{B}$ are {\em  non-uniform} and depend on the product of the norms of the lateral slices of $\T{A}$ and/or the corresponding horizontal slices of $\T{B}$. Following \cite{drineas2006fast}, we consider two examples of nonuniform sampling probabilities:
\begin{description}
  \item[i.]
If we would like to use information from both tensors $\T{A}$ and $\T{B}$, we consider sampling probabilities $\left\{ p_i \right\}_{i=1}^{n_2}$ such that
\begin{equation}
\label{eqn:nearly_optimal_probs}
p_i \geq \beta \frac{\|\T{A}_{:i:}\|_F\|B_{i::}\|_F}
                    {\sum_{i=1}^{n_2} \|{\T{A}_{:i:}}\|_F \|\T{B}_{i::}\|_F}, \qquad  \beta \in (0,1].
\end{equation}
  \item[ii.] If only information on $\T{A}$ is easily available, then we use sampling probabilities $\left\{ p_i \right\}_{i=1}^{n_2}$ such
that
\begin{equation}
\label{eqn:nonoptimal_probs}
p_i \geq \beta \frac{\|{\T{A}_{:i:}}\|_F^2}{\|\T{A}\|_F^2}, \qquad \beta \in (0,1].
\end{equation}
\end{description}

The following theorem gives the main result for Algorithm~\ref{alg:matrix multiply}, and generalizes Theorem~6 in \cite{drineas2008relative} to third order tensors .

\begin{theorem}\label{theorem:matmult-main-exact}
Suppose $\T{A} \in \mathbb{R}^{n_1 \times n_2 \times n_3}$, $\T{B} \in \mathbb{R}^{n_2 \times n_4 \times n_5}$, and $c \le n_2$. In Algorithm~\ref{alg:matrix multiply}, if the sampling probabilities $\left\{ p_i \right\}_{i=1}^{n_2}$ satisfy (\ref{eqn:nearly_optimal_probs}) or (\ref{eqn:nonoptimal_probs}), then the following holds 
\begin{equation}\label{eqn1:theorem:matmult-main-expected}
\Ec[\|{\T{A}*\T{B}-\T{C}*\T{R}}\|_F] \leq \frac{1}{\sqrt{\beta c}}\|{\T{A}}\|_F \|{\T{B}}\|_F.
\end{equation}
\end{theorem}
%
The following lemma establishes a bound for tensor multiplication with respect to the \textit{spectral norm} with improved sampling complexity.

\begin{lemma}\label{lem:mat-mult}
Given a tensor $\T{A}\in\mathbb{R}^{n_1\times n_2 \times n_3}$,  choose $c \geq 48 r\log(4r/(\beta\delta))/(\beta\epsilon^2)$. Let $\T{C} \in\mathbb{R}^{n_1\times c \times n_3}$ and $\T{R}=\T{C}^\top$ be the corresponding subtensors of $\T{A}$.  If sampling probabilities $\left\{ p_i \right\}_{i=1}^{n_2}$ satisfy \eqref{eqn:nonoptimal_probs}, then with probability at least $1-\delta$, we have
\begin{eqnarray}\label{eqn1:theorem:matmult-nor2-expected}
\nonumber
\|\T{A}*\T{A}^\top-\T{C}*\T{C}^\top\| 
&\leq&\max_{k \in [n_3]}\big\{\| \hat{\T{A}}_{::k}\hat{\T{A}}_{::k}^\top- \hat{\T{C}}_{::k}\hat{\T{C}}_{::k}^\top\|_2\big\} \\
&\leq&\epsilon/2 \|\hat{\T{A}}_{::k_\pi}\|_2^2,
\end{eqnarray}
where   $k_\pi =k \in [n_3]$ is the index of the tensor's frontal slice with the maximum spectral norm $\| \hat{\T{A}}_{::k}\hat{\T{A}}_{::k}^\top- \hat{\T{C}}_{::k}\hat{\T{C}}_{::k}^\top\|_2$.
\end{lemma} 

\subsection{Slice-based tensor CX decomposition}

Next, using the concept of free submodules and Definition~\ref{def:lin_ind}, we introduce the novel notion of {\em slice selection} and {\em projection} before formulating a tensor CX decomposition.

\begin{definition}
Let $\T{C} \in \mathbb{R}^{n_1\times c \times n_3}$ and $ \T{X} \in \mathbb{R}^{n_1 \times n_2 \times n_3}$. Then, the tensor project (t-project) operator for $\T{X}$, $ \Pi_{\T{C}}(\T{X})$, is an $n_1 \times n_2 \times n_3 $ tensor obtained as
\begin{equation}\label{tproj}
  \Pi_{\T{C}}(\T{X}):= \T{C} \ast \T{C}^\dagger \ast \T{X},
\end{equation}
where $ \ast$  denotes the t-product, $\Pi_{\T{C}}(\T{X})$ is the projection of $\T{X}$ onto the subspace spanned by the lateral slices of $\T{C}$, and $\T{C}^\dagger$ is the Moore-Penrose generalized inverse of tensor $\T{C}$
\cite{kilmer2013third}.
\end{definition}

Because the circulant convolution of two tube fibers can be computed by the DFT, the t-project can be alternatively computed in the Fourier domain, as shown in Algorithm~\ref{algtproj}.

\begin{algorithm}[!h]
\caption{t-project $\Pi_{\T{C}}(\T{X})$ computation in the Fourier domain.}\label{algtproj}
\begin{algorithmic}[1]
\STATE{\textbf{Input}: $\T{X} \in \mathbb{R}^{n_1 \times n_2 \times n_3}$; $\T{C} \in \mathbb{R}^{n_1 \times c \times n_3}$.}\\
\STATE{$\hat{\T{X}} \leftarrow \texttt{fft}(\T{X}, [], 3);$}
\STATE{$\hat{\T{C}} \leftarrow \texttt{fft}(\T{C}, [], 3);$}
\FOR {$k=1, \dots, n_3$}
\STATE{$ \hat{\T{Z}}_{::k} = \hat{\T{C}}_{::k}\hat{\T{C}}_{::k}^\dagger \hat{\T{X}}_{::k};$}
\ENDFOR\\
\STATE{$\T{Z} \leftarrow \texttt{ifft}(\hat{\T{Z}}, [], 3);$}
\RETURN{$\Pi_{\T{C}}(\T{X})=\T{Z}$}
\end{algorithmic}
\end{algorithm}

\begin{definition}\label{dtcx}
(\textbf{t-CX}). Given an $n_1 \times n_2 \times n_3$ tensor $\T{X}$, let $\T{C}$  be an $n_1 \times c \times n_3$ tensor whose lateral slices correspond to $c$ lateral slices from tensor $\T{X}$. Then, the $n_1 \times n_2 \times n_3$ tensor
 $$  \Pi_{\T{C}}(\T{X})= \T{C} \ast \T{C}^\dagger \ast \T{X},
 $$
is called a t-CX decomposition of $\T{X}$.
\end{definition}

The following remarks are of interest for the previous definition.

\begin{itemize}
  \item The choice for the number of lateral slices $c$ in the t-CX approximation depends on the application under 
	consideration; neverthelees, we are primarily interested in the case $c \ll n_2$. For example, $c$ could be constant, independent of the dimension of tensor, or logarithmic in the size of $n_2$, or a large constant factor less than $n_2$. 
  \item
 The t-CX decomposition expresses each $\T{X}$ slice in terms of a linear combination of basis slices (see, Definition~ \ref{def:lin_ind}), each of which corresponds to an actual lateral slice of $\T{X}$. Hence, t-CX provides a low-rank approximation to the original tensor $\T{X}$, even though its structural properties are different than those of the t-SVD.
  \item Given a set of lateral slices $\T{C}$, the approximation $\Pi_{\T{C}}(\T{X})= \T{C}* \T{C}^\dagger * \T{X}$
      is the "best" approximation to $\T{X}$ in the following sense
\begin{equation}\label{ell2reg} 
\|\T{X} - \T{C}*(\T{C}^\dagger*\T{X})\|_F = \min_{\T{Y}\in \mathbb{R}^{n_1 \times c \times n_3 }} \|\T{X} - \T{C}* \T{Y}\|_F.
\end{equation}      
\end{itemize}

Next, we provide the t-CX decomposition algorithm and provide its relative error. Algorithm~\ref{alg:algCX} takes as input an $n_1 \times n_2\times n_3$ tensor $\T{A}$, a tubal rank parameter $r$, and an error parameter $\epsilon$. It returns as output an $n_1 \times c \times n_3$ tensor $\T{C}$ comprising of a small number of slices of $\T{A}$. Let $\T{L}= \T{U}*\Sigma*\T{V}^\top$, be a tubal rank-$r$ approximation of tensor $\T{A}$. Central to our proposed randomized algorithm is the concept of sampling slices of the tensor based on a leverage score \cite{drineas2008relative}, defined by
\begin{equation}\label{eq:cur-col-prob}
p_i \geq \frac{\beta}{r n_3}\|\hat{\T{V}}_{i::}\|_F^2, \qquad   \forall \ i \in [n_2],   \qquad \beta \in (0,1].
\end{equation}
where $\hat{\T{V}}=\texttt{fft}(\T{V},[],3)$. The goal is to select a small number of lateral slices of $\T{A}$.  

Note that we define rescaling and sampled tensors $\hat{\T{D}}$ and $\hat{\T{S}}$ in the Fourier domain. This leads to a significant reduction in time complexity of the fast Fourier transform and its inverse. For the case $n_3=1$, the algorithm selects column-row pairs in Algorithm~\ref{alg:algCX} and is identical to the algorithm of \cite{drineas2008relative}.
\begin{algorithm}[!h]
\caption{\textbf{(t-CX)}, a fast Monte-Carlo algorithm for tensor low rank approximation.}\label{alg:algCX}
\begin{algorithmic}[1]
\STATE{\textbf{Input}: $\T{A} \in \mathbb{R}^{n_1 \times n_2 \times n_3}$, $p_i \geq 0, i\in[n_2]$ s.t. $\sum_{i \in [n_2]}p_i=1$, a rank parameter $r$, positive integer $c \leq n_2$.}
\STATE{$\hat{\T{A}} \leftarrow \texttt{fft}(\T{A}, [], 3)$;}
\STATE{Initialize $\hat{\T{S}}$ and $\hat{\T{D}}$ to the all zeros tensors;}
\STATE{$t=1$;}
\FOR{$i=1, \dots, n_2$}
\STATE{Pick $i$ with probability $\min\{1,c p_i\}$;}
\IF{$i$ is picked}
\STATE{$\hat{\T{S}}_{it1} = 1, \quad \hat{\T{S}}_{it2:n_3} =0$;}
\STATE{$\hat{\T{D}}_{tt1} = 1/\min\{1,\sqrt{cp_{i}}\}, \quad \hat{\T{D}}_{tt2:n_3} =0;$}
\STATE{$t = t + 1;$}
\ENDIF
\ENDFOR
\FOR {$k=1, \dots, n_3$}
\STATE{$\hat{\T{Z}}_{::k}=\hat{\T{A}}_{::k}\hat{\T{S}}_{::k}\hat{\T{D}}_{::k}(\hat{\T{A}}_{::k}\hat{\T{S}}_{::k}\hat{\T{D}}_{::k})^\dagger\hat{\T{A}}_{::k}$;}
\ENDFOR
\STATE{$\T{Z} \leftarrow \texttt{ifft}(\hat{\T{Z}}, [], 3);$}
\RETURN $\T{Z}$
\end{algorithmic}
\end{algorithm}

Next, we present a lemma of general interest.  The derived properties aid in obtaining tensor based randomized $\ell_2$ regression and low rank estimation guarantees.

\begin{lemma}\label{lem:eps}
Given a target tensor $\T{A}\in\mathbb{R}^{n_1\times n_2 \times n_3}$, let $\epsilon \in (0, 1]$ and $\T{L}= \T{U}*\Sigma*\T{V}^\top$ be a $\rank_t$-r approximation of $\T{A}$. Let $\Gamma =(\T{V}^\top * \T{S}*\T{D})^\dagger - (\T{V}^\top * \T{S}*\T{D})^\top$. If the sampling probabilities $\left\{ p_i \right\}_{i=1}^{n_2}$ satisfy \eqref{eq:cur-col-prob}, and if $c \geq 48 r\log(4r/(\beta\delta))/(\beta\epsilon^2)$, then with probability at least $1-\delta$, the following hold 
\begin{subequations}
\begin{align}
\rank_t(\T{V}^\top * \T{S}) &= \rank_t(\T{V}) = \rank_t(\T{L}),\label{lem:eps1}\\
\|\Gamma\| &= \|\Sigma_{\T{V}^\top*\T{S}* \T{D}}^{-1}-\Sigma_{\T{V}^\top * \T{S} *\T{D}}\|,\label{lem:eps2}\\
(\T{L}*\T{S}*\T{D})^\dagger &= (\T{V}^\top*\T{S} *\T{D})^\dagger\Sigma^{-1}\T{U}^\top,\label{lem:eps3}\\
\|\Sigma_{\T{V}^\top * \T{S} *\T{D}}^{-1}
&-\Sigma_{\T{V}^\top* \T{S}* \T{D}}\|_2 \leq \frac{\epsilon \sqrt{2}}{2},\label{lem:eps4} 
\end{align}
\end{subequations}
where $\Sigma_{\T{V}^\top*\T{S}* \T{D}}$ is an F-diagonal tensor and contains the $r$ non-zero singular tubes of $\T{V}^\top*\T{S}* \T{D}$.
\end{lemma}

Note that the equation \eqref{lem:eps4} shows that in terms of its singular tubes, the tensor $\T{V}^\top * \T{S} *\T{D}$ -i.e., the slice sampled and rescaled version of $\T{V}$- is almost an orthogonal tensor. A useful property of an orthogonal tensor $\T{V}$ is that $\T{V}^\dagger= \T{V}^\top$ (see, Definition~\ref{orth}). Equation \eqref{lem:eps2} shows that although this property does not hold for $\T{V}^\top * \T{S} *\T{D}$, the difference between $(\T{V}^\top * \T{S}*\T{D})^\dagger $ and $ (\T{V}^\top * \T{S}*\T{D})^\top$ can be bounded.

Using Theorem~\ref{theorem:matmult-main-exact} and Lemma~\ref{lem:eps}, we provide in proposition \ref{prop:regress} a sampling complexity for a tensor based randomized ``$\ell_2$-regression" as in \cite{drineas2008relative}. Indeed, Lemmas~2 and 3 of \cite{drineas2008relative} follow by using the Frobenius norm bound of Theorem~\ref{theorem:matmult-main-exact} and applying Markov's inequality. The claim of Lemma~1 of \cite{drineas2008relative} follows by applying Lemma~\ref{lem:eps}.  If we consider the failing probability $\delta=0.05$, the claims of all three lemmas hold simultaneously with probability at least $ 0.85$ and in the following proposition we condition on this event. The proof follows along similar lines to the proof of Theorem~5 in \cite{drineas2008relative}.
 
\begin{proposition}\label{prop:regress}
Given a target tensor $\T{A}\in\mathbb{R}^{n_1\times n_2 \times n_3}$, let $\T{L}= \T{U}*\Sigma*\T{V}^\top$ be a $\rank_t$-r approximation of $\T{A}$. Choose  $c = O(r\log(r/\beta )/(\beta\epsilon^2))$ slices of $\T{A}$ with their sampling probabilities $\left\{ p_i \right\}_{i=1}^{n_2}$ satisfying \eqref{eq:cur-col-prob}. Then, with probability at least $0.85$, the following holds
\begin{equation}\label{eq:regress}
\|{\T{A}-\T{A}*\T{S}*\T{D}*(\T{L}*\T{S}*\T{D})^\dagger*\T{L}}\|_{\text{F}} \leq (1+\epsilon)\|{\T{A}-\T{A}*\T{L}^\dagger*\T{L}}\|_{\text{F}}.
\end{equation}
\end{proposition}

Let $\T{L}= \T{U}*\Sigma*\T{V}^\top$ be a t-SVD of tensor $\T{L} \in\mathbb{R}^{n_1\times n_2 \times n_3}$. In the remainder of the paper, we use the following two parameters related to the coherence of tensor $\T{L}$,
\begin{eqnarray}\label{eq:coh:num}
\nonumber
\varrho:= \frac{r \mu_0(\T{V})}{n_3},\\
\varrho_c:= \frac{c \mu_0(\T{U}_c)}{n_3},
\end{eqnarray}
 where $\T{U}_c$ is the left singular tensor of $\T{C} \in\mathbb{R}^{n_1\times c \times n_3}$ lateral slices of $\T{L}$. 
 
The following theorem shows that projection based on slice sampling leads to near optimal estimation in tensor regression, when the covariate tensor has small coherence, $\mu_0(\T{V})$.

\begin{theorem}\label{thm:regress-main}
Given $\T{A}\in\mathbb{R}^{n_1\times n_2 \times n_3}$,  let $\T{L}= \T{U}*\Sigma*\T{V}^\top$ be a $\rank_t$-r approximation of $\T{A}$. Choose $c =O(\varrho \log(\varrho)/\epsilon^2)$, and let $\T{A}_c\in\mathbb{R}^{n_1\times c\times n_2}$ be a tensor of $c$ lateral slices of $\T{A}$ sampled uniformly without replacement. Further, let $\T{L}_c\in\mathbb{R}^{n_1\times c\times n_3}$ consist of the corresponding slices of $\T{L}$. Then, with probability at least $0.85$, the following holds
$$\|\T{A} - \T{A}_c* \T{L}_c^\dagger*\T{L}\|_F \leq(1+\epsilon)\|\T{A} - \T{A}*\T{L}^\dagger * \T{L}\|_F.$$
\end{theorem}
%
%
Next, by using Lemma~\ref{lem:eps}, we provide a low rank estimation bound for the t-CX algorithm, under the coherence assumption. Indeed, we will take advantage of a constant $\beta$ in equation~\eqref{eq:cur-col-prob} to provide relative error estimation guarantees for the t-CX algorithm under uniform slice sampling, when tensor $\T{A}$ under consideration exhibits
sufficient incoherence.    

\begin{corollary}\label{cor:proj-main}
Given a target tensor $\T{A}\in\mathbb{R}^{n_1\times n_2 \times n_3}$, let $\T{L}= \T{U}*\Sigma*\T{V}^\top$ be a $\rank_t-r$ approximation of $\T{A}$. Choose $c =O (\varrho \log(\varrho) \log(1/\delta)/\epsilon^2)$, and let $\T{C}$ be a tensor of $c$ lateral slices of $\T{A}$ sampled uniformly without replacement. Then, the following holds
\begin{equation}\label{eq:relative:t-CX}
\|\T{A} - \T{C}* \T{C}^\dagger *\T{A}\|_F \leq
(1+\epsilon)\|\T{A} - \T{L}\|_F,
\end{equation}
with probability at least $1-\delta$.
\end{corollary}

\subsection{Slice-based tensor CUR decomposition}

Similar to the matrix case, the t-CX decomposition suffices when $n_1 \ll n_2$ because $\T{C}^\dagger *\T{A}$ is small in size. However, when $n_1$ and $n_2$ are almost equal, computing and storing the dense matrix $\T{C}^\dagger *\T{A}$ in memory becomes prohibitive. The CUR decomposition provides a very useful alternative. Next, we introduce a tensor CUR (t-CUR) decomposition based on the rt-product.

\begin{definition}\label{dtcur} (\textbf{t-CUR}).
Let $\T{X}$ be an $n_1 \times n_2 \times n_3$ tensor. For any given $\T{C}$, an $n_1 \times c \times n_3$ tensor whose slices comprise of $c$ lateral slices of tensor $\T{X}$, and $\T{R}$, an $l \times n_2 \times n_3$ tensor whose slices comprise of $l$ horizontal slices of tensor $\T{X}$, the $n_1 \times n_2 \times n_3$ tensor
$$
\T{C} \ast \T{U} \ast \T{R},
 $$
is a lateral-horizontal-based tensor approximation to $\T{X}$ for any $ c \times l \times n_3$ tensor $\T{U}$.
\end{definition}

The following remarks are of interest regarding the previous definition.
\begin{itemize}
\item The t-CUR decomposition is most appropriate as a data analysis tool, when the data consist of one- and/or two modes that are qualitatively different than the remaining ones. In this case, the t-CUR decomposition approximately expresses the original data tensor in terms of a basis consisting of underlying subtensors that are actual data slices and not artificial bases.
  \item The t-CUR approximation is a t-CX approximation, but one with a very special structure; i.e., every lateral slice of $\T{X}$ can be expressed in terms of the basis provided by $\T{C}$ (see, Definition~\ref{def:lin_ind}) using only the information contained in a small number of horizontal slices of $\T{X}$ and a low-dimensional encoding tensor.
  \item In terms of its t-SVD structure, $\T{U}$ contains the "inverse-of-$\T{X}$" information. For the proposed t-CUR decomposition, $\T{U}$ will be a generalized inverse of the intersection between selected tensors $\T{C}$ and $\T{R}$.
\end{itemize}

Note that the structural simplicity of the t-CUR tensor decomposition becomes apparent in the Fourier domain, as detailed in Algorithm~\ref{alg:algCUR}. The latter takes as input an $n_1 \times n_2 \times n_3$ tensor $\T{A}$, an $n_1 \times c \times n_3$ tensor $\T{C}$ consisting of a small number of lateral slices of $\T{A}$, and an error parameter $\epsilon$. Letting $\hat{\T{C}}= \hat{\T{U}}*\hat{\Sigma}*\hat{\T{V}}^\top,$ Algorithm~\ref{alg:algCUR} uses sampling probabilities
\begin{equation}\label{eqn:colprob_rows_exact}
p_i \geq \frac{\beta}{cn_3}\|\hat{\T{U}}_{i::}^\top\|,
\qquad   \forall i \in [n_1],   \qquad \beta \in (0,1],
\end{equation}
to select a small number of horizontal slices of $\T{A}$.
It returns an $l \times n_2 \times n_3 $ tensor $\T{R}$ comprising of a small number of horizontal slices of $\T{A}$ and an $c \times l \times n_3$ tensor $\T{U}$ consisting of the corresponding slices of $\T{C}$.

\begin{algorithm}[!h]
\caption{\textbf{(t-CUR)}, a fast Monte-Carlo algorithm for tensor low rank approximation.}\label{alg:algCUR}
\begin{algorithmic}[1]
\STATE{\textbf{Input}: $\T{A} \in \mathbb{R}^{n_1 \times n_2 \times n_3}$, $\T{C} \in \mathbb{R}^{n_1 \times c \times n_3}$ consisting of $c$ lateral slices of $\T{A}$, $p_i \geq 0, i\in[n_2]$ s.t. $\sum_{i \in [n_2]}p_i=1$, a rank parameter $r$, and positive integer $l \leq n_1$.}
\STATE{$\hat{\T{A}} \leftarrow \texttt{fft}(\T{A}, [], 3)$;}
\STATE{$\hat{\T{C}} \leftarrow \texttt{fft}(\T{C}, [], 3)$;}
\STATE{Initialize $\hat{\T{S}}$ and $\hat{\T{D}}$ to the all zeros tensors;}
\STATE{$t=1$;}
\FOR{$i=1, \dots, n_1$}
\STATE{Pick $i$ with probability $\min\{1,lp_i\}$;}
\IF{$i$ is picked}
\STATE{$\hat{\T{S}}_{it1} = 1, \quad \hat{\T{S}}_{it2:n_3} =0$;}
\STATE{$\hat{\T{D}}_{tt1} = 1/\min\{1,\sqrt{lp_{i}}\}, \quad \hat{\T{D}}_{tt2:n_3} =0;$}
\STATE{$t = t + 1;$}
\ENDIF
\ENDFOR
\FOR {$k=1, \dots, n_3$}
\STATE{$\hat{\T{R}}_{::k}=\hat{\T{D}}_{::k}\hat{\T{S}}_{::k}^\top \hat{\T{A}}_{::k}$;}
\STATE{$\hat{\T{U}}_{::k}= \big(\hat{\T{D}}_{::k}\hat{\T{S}}_{::k}^\top \hat{\T{C}}_{::k}\big)^\dagger$;}
\STATE{$\hat{\T{Z}}_{::k}=\hat{\T{C}}_{::k}\hat{\T{U}}_{::k}\hat{\T{R}}_{::k}$;}
\ENDFOR
\STATE{$\T{Z} \leftarrow \texttt{ifft}(\hat{\T{Z}}, [], 3);$}
\RETURN  $\T{Z}$
\end{algorithmic}
\end{algorithm}

\begin{corollary}
\label{cor:gnys-main}
Given $\T{A}\in\mathbb{R}^{n_1\times n_2 \times n_3}$,  let $\T{L}= \T{U}*\Sigma*\T{V}^\top$ be a $\rank_t-r$ approximation of $\T{A}$. In Algorithm~\ref{alg:algCUR}, choose $c =O (\varrho \log(\varrho) \log(1/\delta)/\epsilon^2)$, and let $\T{C} \in \mathbb{R}^{n_1\times c \times n_3}$ be a tensor of $c$ slices of $\T{A}$ sampled uniformly without replacement. Further, choose  $l =O (\varrho_c\log(\varrho_c)\log(1/\delta)/\epsilon^2)$, and let $\T{R} \in \mathbb{R}^{l \times n_2 \times n_3}$ be a tensor of $l$ horizontal slices of $\T{A}$ sampled uniformly without replacement. Then, the following holds
$$\|\T{A} - \T{C}*\T{U}*\T{R}\|_F \leq (1+\epsilon)\|\T{A} - \T{L}\|_F,$$
with probability at least $1-\delta$.
\end{corollary}

\begin{remark}
In many applications such as tensor completion problems discussed in Section \ref{result:completion}, it may not feasible to compute the t-SVD of the entire tensor $\T{A}$ due to either computational cost or large set of missing values. In these cases, algorithms requiring knowledge of the leverage scores can not be applied. Hence, we may use an estimate where a subset of the lateral slices are chosen uniformly and without replacement, and the horizontal leverage scores \eqref{eqn:colprob_rows_exact} are calculated using the top-$r$ right singular slices of this lateral tensor instead of the entire $\T{A}$ tensor. 
\end{remark}

\begin{remark}
Note that for any subspace, the smallest $\mu_{0}$ can be is 1. In such case, we are using the optimal subspace sampling with $\beta=1$. Note also that we have provided Corollaries~\eqref{cor:proj-main} and \eqref{cor:gnys-main} based on uniform sampling. However, it can be easily seen that the relative error guarantees hold with nonuniform sampling probabilities \eqref{eq:cur-col-prob} and \eqref{eqn:colprob_rows_exact}, if we set $c =O (r \log(r/\beta) \log(1/\delta)/(\beta\epsilon^2))$, and $l =O (c \log(c/\beta) \log(1/\delta)/(\beta\epsilon^2))$. 
\end{remark}

\section{Tensor completion and robust factorization}\label{noisy:completion}

Next, we provide a CUR tensor nuclear norm minimization (CUR t-NN) procedure that solves a noisy tensor factorization problem, using a small number of lateral and/or horizontal slices of the underlying tensor and exhibits favorable computational complexity and in addition comes with performance guarantees.

\subsection{Related work on matrix problems}
\noindent
{\em Low rank plus sparse matrix decomposition.}
In many image processing and computer vision applications, the given data matrix $X$ can be decomposed as a sum of a low-rank and a sparse component. To that end, Cand\'{e}s et al. proposed Robust PCA \cite{candes2011robust} to model data matrices generated according to the following mechanism:
\begin{equation}\label{pca}
  \min_{L} \rank(L) +\|E\|_{0}   \qquad s.t.  \qquad X=L+E.
\end{equation}
Note that the solution of \eqref{pca} is an NP-hard problem. It is established in \cite{candes2009exact} that if $L$ exhibits a certain degree of low-rankness, while $E$ is sparse enough, then the formulation of \eqref{pca} can be relaxed into a convex problem of the form:
\begin{equation}\label{rpca}
  \min_{L} \|L\|_* +\|E\|_{1}   \qquad s.t.  \qquad X=L+E.
\end{equation}

This model implicitly assumes that the underlying data structure lies in a single low-rank subspace. However, in many applications (e.g. image classification) it is more likely that the data are obtained from a union of multiple subspaces, and hence recovery of the structure based on the above decomposition would be inaccurate. In order to segment the data into their respective subspaces, one needs to compute an affinity matrix that encodes the pairwise affinities between data vectors. Liu \cite{liu2013robust} proposed a more general rank minimization problem, where the data matrix itself is used as the dictionary, resulting in the following convex optimization problem:
\begin{equation}\label{opca}
  \min_{L} \|L\|_{*} +\|E\|_{2,1}   \qquad s.t.  \qquad X= XL+E, \qquad  \diag(L)=0.
\end{equation}

When the subspaces are \emph{globally} independent, the data are noiseless and sampling is sufficient, Liu et al. \cite{liu2013robust} show that the optimal solution, denoted by $L^*$, to the problem given by \ref{opca} corresponds to the widely used Shape Iteration Matrix (SIM) method \cite{costeira1998multibody}. The latter is a  "block-diagonal" affinity matrix that indicates the true segmentation of the data. To handle data corrupted by noise, the popular Low Rank Representation (LRR) introduced in \cite{liu2013robust}
adopts a regularized formulation that introduces an extra penalty term to fit the noise component. Further, after obtaining the self-representation matrix $L$, the
affinity matrix $C$ is usually constructed as $C= \frac{1}{2} (|L|+|L^\top|),$ where $|\cdot|$ represents the absolute operator. Then, the obtained affinity matrix $C$ will be processed through a spectral clustering algorithm \cite{ng2002spectral} to produce the final clustering result and obtain the corresponding data generating subspaces.

It is established in \cite{zhang2013learning, wang2013provable} that combining sparse and low-rank regularization can improve the performance of image classification. The basic objective function of this combination \cite{zhang2013learning} is as follows:
\begin{equation}\label{spalow}
  \min_{L} \lambda_1\|L\|_{*} + \lambda_2\|L\|_1 + \lambda_3 \|E\|_{\ell}   \qquad s.t.  \qquad X= DL+E, \qquad  \diag(Z)=0,
\end{equation}
where $\lambda_1,\lambda_2, \lambda_3$ are tuning parameters and $\|E\|_{\ell}$ indicates different norms suitable for different types for corrupting the data by noise; for example, the squared Frobenius norm for Gaussian noise and the $\ell_1$ norm for random spiked noise. Equation \eqref{spalow} is similar to the objective functions in \cite{wang2013provable,foggia2014graph}, where a detailed explanation of the formulation given in \eqref{spalow} is also provided.

\subsubsection{Low rank tensor decomposition}

Lu et al. \cite{lutensor} extended Robust PCA \cite{candes2011robust} to the third order tensor based on t-SVD and proposed the following convex optimization problem:
\begin{equation}\label{tlow7}
  \min_{\T{L}, \T{E}} \|\T{L}\|_{\circledast} + \lambda\|\T{E}\|_{1} \qquad s.t.  \qquad
  \T{X}=  \T{L}+ \T{E},
\end{equation}
where $\lambda$ is a tuning parameter. 

This model implicitly assumes that the underlying data come from a single low-rank subspace. When the data is drawn from a union of multiple subspaces, which is common in image classification, the recovery based on the above formulation may lack
in accuracy. To that end, Xie et al. \cite{xie2016multi} extended LRR based subspace clustering to a multi-view one, by employing the rank sum of different mode unfoldings to constrain the subspace coefficient tensor, resulting in the following convex optimization problem:
\begin{equation}\label{tlow8}
  \min_{\T{L}}  \|\T{L}\|_{\circledast} + \lambda \|\T{E}\|_{2,1} \qquad s.t.  \qquad
   \T{X}= \T{X} * \T{L} + \T{E},
\end{equation}

Equation \eqref{tlow8} is similar to the objective functions in \cite{kernfeld2014clustering, piao2016submodule} and a detail explanation of \eqref{tlow8} is provided in that paper.

\subsection{Proposed algorithm}\label{sec:admm:cur}

Next, we propose an algorithm for large scale tensor decomposition of noisy data. Our proposal, called CUR t-NN, extends Algorithm~\ref{alg:algCUR} to the \textit{noisy} tensor factorization. The main steps are outlined in Algorithm~\ref{CUR t-NN}.

Note that CUR t-NN can be used in combination with an arbitrary optimization algorithm.  In this paper, we have chosen to solve the noisy tensor factorization formulations \eqref{tlow6}, \eqref{tlow7} and \eqref{tlow8} using an Alternating Direction Method of Multiplier (ADMM) algorithm \cite{boyd2011distributed}.  ADMM is the most widely used approach for robust tensor PCA in both the fully and partially observed settings. Indeed, ADMM achieves much higher accuracy than (accelerated) proximal gradient algorithm using fewer iterations. It works well across a wide range of problem settings and does not require  
careful tuning of the regularization parameters. Further, the following empirical finding has been frequently observed:
namely, the rank of the iterates often remains bounded by the rank of the initializer, thus enabling efficient
computations \cite{candes2011robust}. This feature is not shared by the block coordinate decent algorithm.  We provide a variant of ADMM for solving problem~\eqref{tlow7} in the Appendix (see, Algorithm~\ref{alg1}). With a small modification, Algorithm~\ref{alg1} can be also used to solve the tensor completion problem~\ref{tlow6}, and tensor subspace clustering problem~\eqref{tlow8}. In the following algorithm, $\textbf{ADMM} (\T{X} , \lambda)$ denotes the ADMM algorithm for solving regularized tensor nuclear norm minimization problems~\eqref{tlow6}, \eqref{tlow7} and \eqref{tlow8}, where $\T{X}$ is the  (sampled) data tensor and $\lambda$ is a regularization parameter. 

\begin{algorithm}[h]
\caption{\texttt{CUR t-NN}, tensor nuclear norm minimization based on CUR factorization.}\label{CUR t-NN}
\begin{algorithmic}[1]
\STATE{$\T{X} \in \mathbb{R}^{n_1 \times n_2 \times n_3}$, positive integer $c$ and $l$, and a regularization parameter $\lambda$;}
\STATE{Let $\T{C} $ be a tensor of $c$ selected lateral slices of $\T{X}$ using probabilities \eqref{eq:cur-col-prob};}
\STATE{ $\tilde{\T{C}} \leftarrow \textbf{ADMM} (\T{C} , \lambda)$; } 
\STATE{ Let $\T{R} $ be a tensor of $l$ selected horizontal slices of $\T{X}$ using probabilities \eqref{eqn:colprob_rows_exact};}
\STATE{$\tilde{\T{R}} \leftarrow\textbf{ADMM} (\T{R}, \lambda)$; }
\STATE{Let $\tilde{\T{U}} = \tilde{\T{W}}^\dagger$, where $\tilde{\T{W}}$ is the $l \times c \times n_3$ tensor formed by sampling the corresponding $l$ horizontal slices of $\tilde{\T{C}}$;}
\STATE{$\tilde{\T{L}}= \tilde{\T{C}} * \tilde{\T{U}} *\tilde{\T{R}}$;}
\RETURN $\tilde{\T{L}}$
\end{algorithmic}
\end{algorithm}

\subsection{Running Time of CUR t-NN}

Algorithm~\ref{CUR t-NN} significantly reduces the per-iteration complexity of nuclear norm minimization problems. Indeed, in each iteration, a base tensor nuclear norm minimization algorithm requires $O(n_1 n_2 n_3 \log(n_3))$ flops to transform the tensor to the Fourier domain, $O(n_1 n_2 n_3 \min(n_1,n_2))$ flops for the t-SVD computation and factorization, and $O(n_1 n_2 n_3 \log(n_3))$ flops to transform it back to the original domain. On the other hand, Algorithm~\ref{CUR t-NN} only requires $O(\max(c n_1, l n_2) n_3 \log(n_3))$, $O(\max(n_1 c, l n_2) n_3 r)$ and $O(\max(c n_1, l n_2) n_3 \log(n_3))$ flops, for the respective steps. Further, Algorithm~\ref{CUR t-NN} can be implemented without storing the data tensor $\T{X}$ and can be advantageous when $r \ll \min(n_1, n_2)$, which occurs frequently in real data sets.

\subsection{Theoretical guarantees}

Next, using Lemma~\ref{lem:eps}, we establish that Algorithm~\ref{CUR t-NN} exhibits high probability recovery guarantees comparable to those of the exact algorithms that use the full data tensor. Our first result bounds the $\mu_0$ and $\mu_1$-coherence (see, Definitions \eqref{eq:ten:incoh} and \eqref{eq:ten:incoh1}) of a randomized tensor in terms of the coherence of the full tensor.

\begin{lemma}\label{lem:sub-coh}
Let $\T{L}_c$ be a tensor formed by selecting $c$ slices of a $\rank_t-r$ tensor $\T{L} =\T{U} *\Sigma * \T{V}$ that satisfy the probabilistic conditions in \eqref{eq:cur-col-prob}. If $c \geq 48 r\log(4r/(\beta\delta))/(\beta\epsilon^2)$, then with probability at least $1-\delta$:
\begin{itemize}
\item  $\mu_0(\T{U}_{\T{L}_c}) = \mu_0(\T{U})$,
\item $\displaystyle\mu_0(\T{V}_{\T{L}_c}) \le \frac{1}{1-\epsilon/2}\mu_0(\T{V})$,
\item  $\displaystyle\mu_1(\T{L}_c ) \le  \frac{r}{1-\epsilon/2}\mu_0(\T{U})\mu_0(\T{V})$,
\end{itemize}
where $\epsilon \in (0,1]$.
\end{lemma}
 
Our next theorem provides a bound for the estimation error of the CUR t-NN Algorithm. 

\begin{theorem}\label{thm:master}
Under the notion of Algorithm~\ref{CUR t-NN}, choose $c =O (\varrho \log(\varrho) \log(1/\delta)/\epsilon^2)$, and $l =O (\varrho_c\log(\varrho_c)\log(1/\delta) /\epsilon^2)$. Let $\T{C}^*$ and $\T{R}^*$ be the corresponding lateral and horizontal sub-tensors of the exact solution $\T{L}^*$. If $\T{L}^*$ is $(\mu, r)$-coherent, then with probability at least $1-\delta$, $\T{C}^*$ and $\T{R}^*$ are $(\frac{r \mu^2}{1-\epsilon/2},r)$-coherent, and 
$$
\|\T{L}^* - \tilde{\T{L}}\|_F \leq (2+\epsilon) \sqrt{
\|\T{C}^*-\tilde{\T{C}}\|_F^2+\|\T{R}^*-\tilde{\T{R}}\|_F^2},
$$
where $\tilde{\T{L}}$ is a solution obtained by Algorithm~\ref{CUR t-NN}.   
\end{theorem}

\section{Experimental Results}

Next, we investigate the efficiency of the proposed randomized algorithms on both synthetic and real data sets. The results are organized in the following three sub-sections: in Section~\ref{result:multi}, we compare the performance of the rt-product to that of the t-product using synthetic data. In Section~\ref{result:brain}, we use the proposed {t-CX} and {t-CUR} algorithms for finding important ions and positions in two Mass Spectrometry Imaging (MSI) data sets.
Finally, in Section~\ref{result:completion}, we apply the proposed  tensor factorization {CUR t-NN} algorithm on data sets related to image and video processing.

All algorithms have been implemented in the MATLAB R2018a environment and run on a Mac machine equipped with a 1.8 GHz Intel Core i5 processor and 8 GB 1600 MHz DDR3 of memory. 

\subsection{Experimental results for fast tensor multipication}\label{result:multi}

A random data tensor $\T{X} \in \mathbb{R}^{n_1 \times n_2 \times n_3 }$ of dimension $n_1 =10^5$, $n_2=10^3$ and $n_3 =5$
is generated as follows: the first frontal slice of $\T{X}$ is set to $\T{X} (:,:,1) =\texttt{sprandn} (n_1,n_2,1)$; i.e. comprising of $n_1\times n_2$ sparse normally distributed random matrix. The remaining frontal slices are generated by $\texttt{sprandn} (\T{X} (:,:,1))$, which results in the {\em same sparsity} structure as $\T{X} (:,:,1)$, with normally distributed random entries with mean 0 and variance 1. 

We perform a t-SVD decomposition on $\T{X}$ and select the left singular value tensor $\T{U}$ for our experiments. We consider both uniform and nonuniform sampling schemes. For nonuniform sampling, we use a randomized approach (similar to Algorithm 4 in \cite{drineas2012fast}) to obtain normalized leverage scores of $\T{U}$. Further, the Frobenius and spectral
bounds given in \eqref{eqn1:theorem:matmult-main-expected} and \eqref{eqn1:theorem:matmult-nor2-expected}, respectively,
are used as performance metrics for Algorithm~\ref{alg:matrix multiply}. Table~\eqref{table:ion} shows the Relative Frobenius Error (RFE) -$\|\T{U}_r^ \top \T{U}_r -\tilde{\T{U}}_r^\top\tilde{\T{U}}_r \|_F/ \|\T{U}_r\|_F^2$- and  Relative Spectral Error (RSE) -$\|\T{U}_r^ \top \T{U}_r -\tilde{\T{U}}_r^\top\tilde{\T{U}}_r \|_2/ \|\T{U}_r\|_2^2$- for the rt-product algorithm to recover $\tilde{\T{U}}_r^\top \tilde{\T{U}}_r =\T{U}_r^\top *\T{S} *\T{D} *\T{D}* \T{S}^\top*\T{U}_r$ based on $ l \approx r \log(r)$ slices, and its deterministic counterpart. The depicted results are averaged over 10 independent replications.
\begin{table}[h]
\centering
\rowcolors{1}{}{lightgray}
\begin{tabular}{|c|c|ccc|ccc|}
\hline
\multicolumn{1}{|c|}{}&
\multicolumn{1}{c}{t-product}&
\multicolumn{6}{c|}{rt-product}\\
\multicolumn{1}{|c|}{Rank, selected slices}&
\multicolumn{1}{c}{}&
\multicolumn{3}{c}{non-uniform sampling}&\multicolumn{3}{c|}{uniform sampling}\\
\multicolumn{1}{|c|}{} & time (s)&time(s)&RFE&RSE& time(s)&RFE& RSE\\
 \hline
r=100, l=460  &9.09& 8.13& 70e-3&43e-4&  1.03 & 39e-2& 11e-2\\
r=200, l=1000 &36.36  &14.27&64e-3&25e-4  &1.08&12e-2& 82e-2\\
r=300, l=1700 &162.99 &34.89&56e-3&29e-4 &1.39&41e-2& 34e-2\\
r=400, l=2500 &321.43 &41.04&31e-3  &13e-3 &1.64&90e-2& 38e-3\\
r=500, l=3100 &684.01 & 67.07&71e-3& 12e-4 & 3.09&23e-2& 43e-3\\
\hline
\end{tabular}
\caption{Relative errors and running times of tensor multipication algorithms.}\label{table:multi}
\end{table}

Very large improvements in computational speed can be seen when using the {rt-product}, especially when coupled with
a non-uniform sampling scheme. Further, the gains become more pronounced for larger number of slices sampled and larger
rank.

\subsection{Finding Important Ions and Positions in MSI}\label{result:brain}

MSI is used to visualize the spatial distribution of chemical compounds, such as metabolites, peptides or proteins by their molecular masses. The ability of MSI to localize panels of biomolecules in tissue samples has led to a rapid and substantial impact in both clinical and pharmacological research, aided in uncovering biomolecular changes associated with disease and finally provided low cost imaging of drugs. Typical techniques used require finding important ions and positions from a 3D image: ions to be used in fragmentation studies to identify key compounds, and positions for follow up validation measurements using microdissection or other orthogonal techniques. Unfortunately, with modern imaging machines, these must be selected from an overwhelming amount of raw data. Existing popular techniques used to reduce the volume of data, include principal components analysis and non-negative matrix factorization,  but return difficult-to-interpret linear combinations 
of actual data elements. A recent paper \cite{yang2015identifying} shows that CX and CUR matrix decompositions can be used directly to address this selection need. One major shortcoming of CX and/or CUR matrix decompositions is that they can only handle 2-way (matrix) data. However, MSI data form a multi-dimensional array. Hence, in order to use CX/CUR matrix decompositions, one has first to reformulate the multi-way array as a matrix. Such preprocessing usually leads to information loss, which in turn could cause significant performance degradation. 

By using instead the t-CX and t-CUR decompositions (Algorithms~\ref{alg:algCX} and \ref{alg:algCUR}, respectively) one
can obtain good low-rank approximations of the available data, expressed as combinations of actual ions and positions, as opposed to difficult-to-interpret eigen-ions and eigen-positions produced by matrix factorization techniques.
We show that this leads to effective prioritization of information for both actual ions and actual positions. In particular, important ions can be discovered by using leverage scores as the importance sampling distribution. Further, selection of important positions from the original tensor can be accomplished based on the random sampling algorithm in \cite{yang2015identifying}, since the distribution of the leverage scores of positions is uniform. To this end, we consider the following two ways of computing leverage scores of a given data set:
\begin{itemize}
  \item Deterministic: Compute the normalized tensor leverage scores exactly using probabilities \eqref{eq:cur-col-prob} and \eqref{eqn:colprob_rows_exact}.
  \item Randomized: Compute an approximation to the normalized leverage scores of tensor (mapped to a  block diagonal matrix in the Fourier domain) by using Algorithm~4 of \cite{drineas2012fast}.
  \end{itemize}

\subsubsection{Description and Analysis of MSI data sets}

Next, we use the following two data sets for illustration purposes, that are publicly available at the OpenMSI Web gateway\footnote{\texttt{https://openmsi.nersc.gov/openmsi/client/}}. They represent two diverse acquisition modalities, including one mass spectrometry image of the left coronal hemisphere of a mouse brain (see, Figure~\ref{fig:msi}) acquired using a time-of-flight (TOF) mass analyzer and one MSI data set of a lung acquired using an Orbitrap mass analyzer. These data sets form
a $122 \times 120 \times 80339 $ and a $122 \times 120 \times 500000 $ tensor, respectively.
\begin{figure}[h]
	\centering
			\includegraphics[scale=0.8]{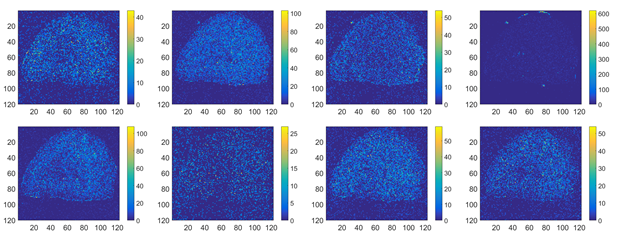}
\caption{Sample ions in the brain data set.}\label{fig:msi}		
\end{figure}
As described in \cite{yang2015identifying}, the brain data set is processed using peak-finding to identify the most intense ions. Using this technique, the original brain data is reduced from $122 \times 120 \times 80339 $ values to a data set of size $122 \times 120 \times  2000$. To compute the CX decomposition, we reshape the MSI data cube into a two-dimensional $(14640 \times 2000)$ matrix, where each row corresponds to the spectrum of a pixel in the image, and each column to the intensity of an ion across all pixels, thus describing the distribution of the ion in physical space. No peak-finding was applied to the lung data set.
\begin{figure}[h]
	\centering \makebox[0in]{
		\begin{tabular}{cc}
			\includegraphics[scale=0.2]{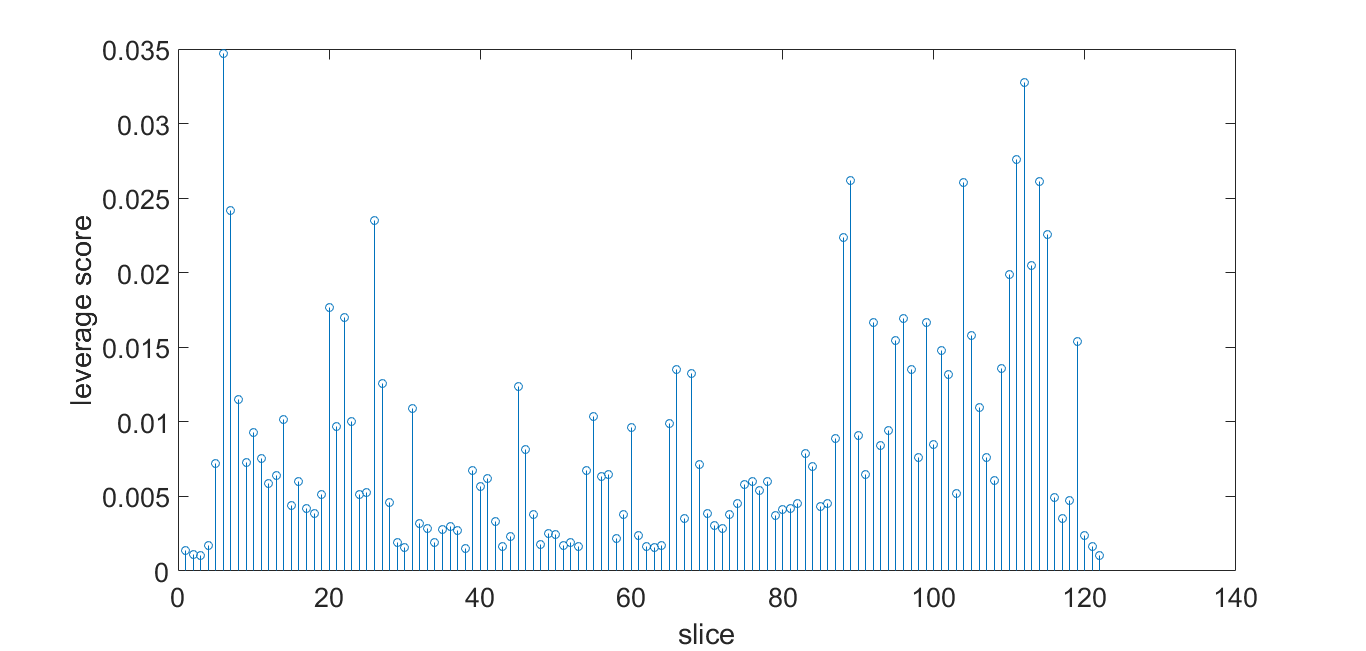}
			\includegraphics[scale=0.2]{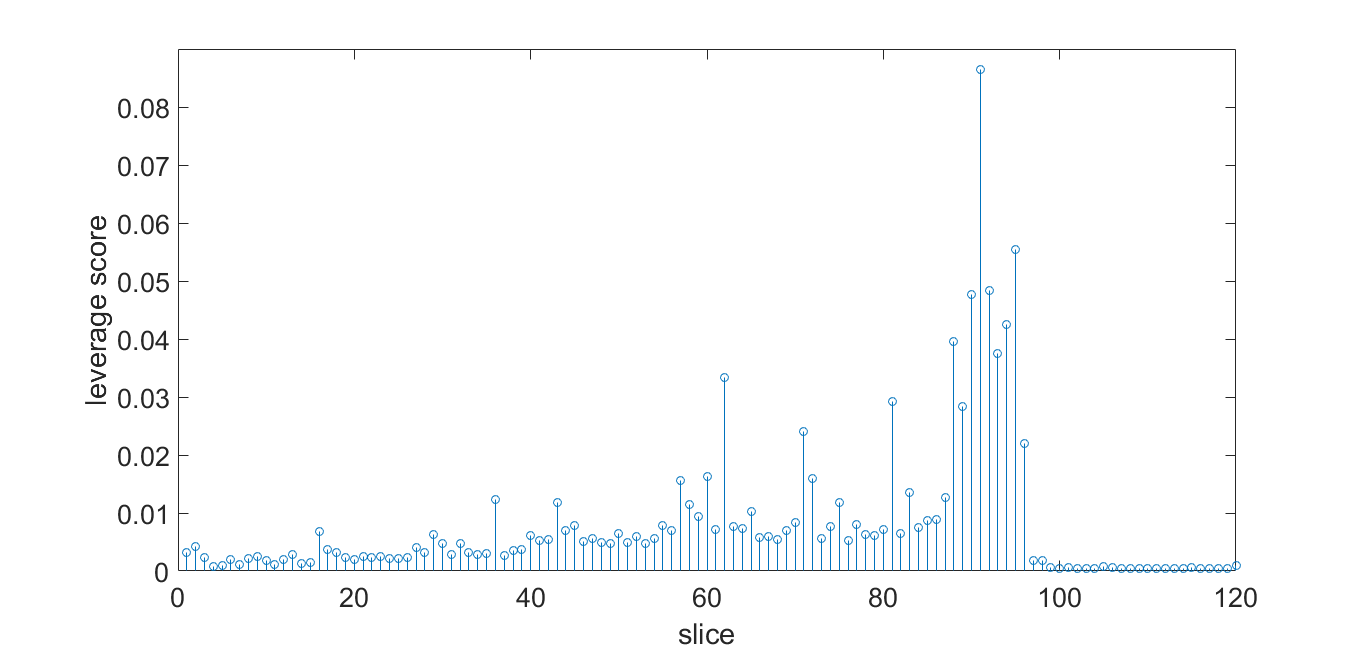}\\
			\includegraphics[scale=0.3]{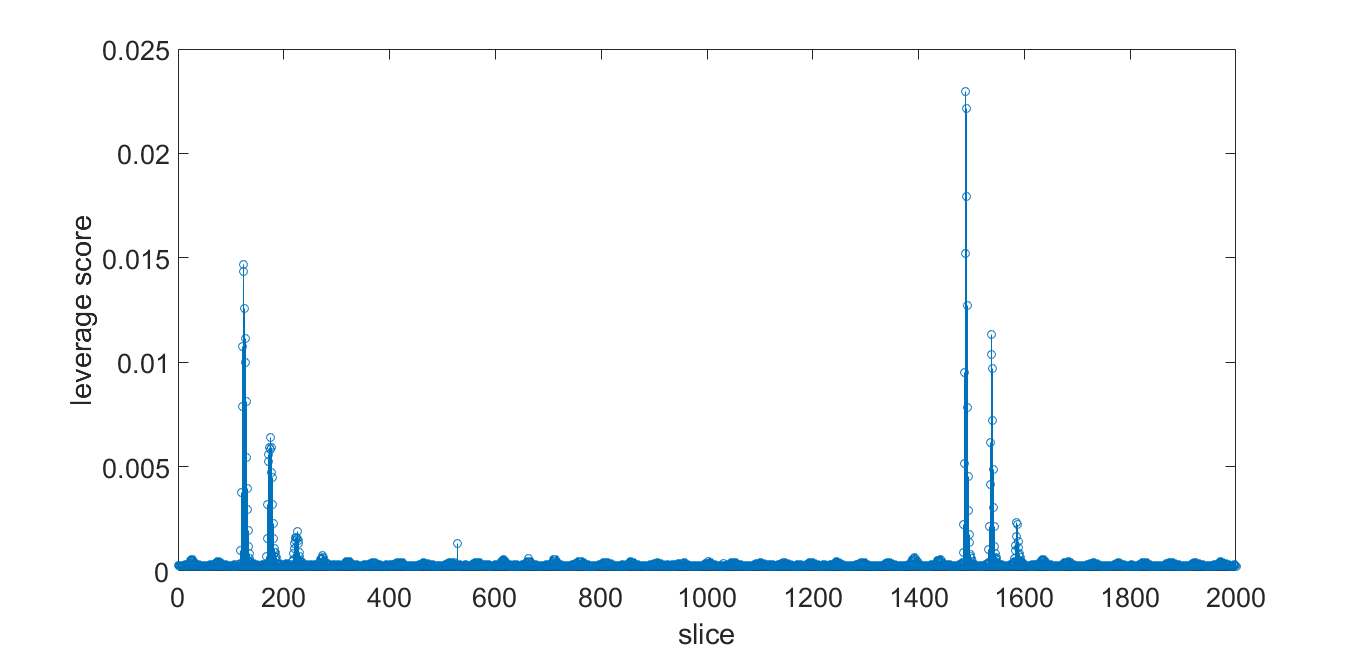}
		\end{tabular}}
\caption{Distribution of leverage scores of tensor $\T{X}$, relative to the best rank$-5$ space for the brain data set. \textbf{Left}: Horizontal scores; \textbf{Right}: Lateral scores;  \textbf{Bottom}: Frontal scores.}\label{fig:scores}
\end{figure}

\begin{table}[h]
\centering
\rowcolors{1}{}{lightgray}
\begin{tabular}{|c|cccccccc|}
\hline
\multicolumn{1}{|c|}{Number of}&\multicolumn{4}{c}{CX}&\multicolumn{4}{c|}{t-CX}\\
\multicolumn{1}{|c|}{selected slices}&\multicolumn{2}{c}{Randomized}&\multicolumn{2}{c}{Deterministic}&\multicolumn{2}{c}{Randomized}&\multicolumn{2}{c|}{Deterministic}\\
\multicolumn{1}{|c|}{} &  RSE&time&RSE&time&RSE&time& RSE&time\\
 \hline
25 & 19.13e-2& 4.56 &18.84e-2&8.46&13.65e-2&4.13& 16.05e-2&4.47\\
35 & 17.24e-2& 5.12 &17.59e-2& 8.95&17.43e-2&4.87 &13.13e-2&5.99\\
45 & 16.35e-2& 7.01 &16.93e-2& 14.99&11.32e-2& 6.14&11.01e-2&7.01\\
55 & 15.14e-2& 8.16 &16.52e-2& 15.86& 16.26e-2& 6.63&10.16e-2&8.99\\
\hline
\end{tabular}
\caption{RSE of matrix and tensor low rank decomposition relative to the best rank$-5$ space for identifying important ions in the brain data set.}\label{table:ion}
\end{table}

\begin{table}[h]
\centering
\rowcolors{1}{}{lightgray}
\begin{tabular}{|c|cccccccc|}
\hline
\multicolumn{1}{|c|}{Number of}&\multicolumn{4}{c}{CX}&\multicolumn{4}{c|}{t-CX}\\
\multicolumn{1}{|c|}{selected slices}&\multicolumn{2}{c}{Randomized}&\multicolumn{2}{c}{Deterministic}&\multicolumn{2}{c}{Randomized}&\multicolumn{2}{c|}{Deterministic}\\
\multicolumn{1}{|c|}{} &  RSE&time&RSE&time&RSE&time& RSE&time\\
 \hline
25 & 45.24e-2&3.15&46.03e-2&7.41&37.65e-2& 3.01&16.05e-2&4.06\\
35 & 35.87e-2& 3.81&35.68e-2&7.63&28.00e-2&3.19& 13.13e-2&4.55\\
45 & 24.13e-2&5.29&24.19e-2&11.89&23.98e-2&4.15& 21.01e-2&5.59\\
55 & 23.16e-2&5.83&24.39e-2&12.93&15.23e-2&4.23& 15.17e-2&5.71\\
\hline
\end{tabular}
\caption{RSE of matrix and tensor low rank decomposition relative to the best rank$-5$ space for finding important pixels in the brain data set.}\label{table:pix}
\end{table}

For the brain data set, we evaluate the quality of the approximation of the leverage scores based on a rank $r=5$ approximation. The distribution of the deterministic leverage scores for the ions and pixels is shown in Figure~\ref{fig:scores}. Table~\eqref{table:ion} shows the relative square error (RSE) $\|\T{X}- \tilde{\T{X}}\|_F/\|\T{X}\|_F$ using t-CX decomposition to recover rank $-5$ tensor $\tilde{\T{X}}$ for selection of $c = 25, 35, 45$ and, $55$ ions, using both randomized and deterministic leverage scores. Table \eqref{table:pix} provides the reconstruction errors to the best rank$-5$ approximation, based on the t-CX decomposition coupled with horizontal slice selection. The results obtained  running both randomized and deterministic CX and t-CX algorithms are based on 10 independent replicates and then averaging them. Note that for pixel selection, the deterministic CX decomposition results in larger reconstruction errors than its randomized CX counterpart. The reason for this behavior lies in the distribution of the leverage scores for the pixels, which are fairly uniform (see, Figure~\ref{fig:scores}).

\begin{table}[h]
\centering
\rowcolors{1}{}{lightgray}
\begin{tabular}{|c|cccccccc|}
\hline
\multicolumn{1}{|c|}{Number of}&\multicolumn{4}{c}{Decomposition using ions}&\multicolumn{4}{c|}{Decomposition using ions and pixels}\\
\multicolumn{1}{|c|}{selected slices}&\multicolumn{2}{c}{CX}&\multicolumn{2}{c}{t-CX}&\multicolumn{2}{c}{CUR}&\multicolumn{2}{c|}{t-CUR}\\
\multicolumn{1}{|c|}{} &  RSE&time&RSE&time&RSE&time& RSE&time\\
 \hline
25 & 40.82e-2&3.31&46.89e-2&3.05&59.76e-2&2.89&24.67e-2&2.14\\
35 & 41.98e-2&4.13&29.63e-2&3.24&42.57e-2&3.16&22.71e-2&2.67\\
45 & 30.16e-2& 5.78&22.35e-2&3.36&12.18e-2&4.25&18.15e-2&3.13\\
55 & 30.74e-2&6.67&20.81e-2&4.05&19.00e-2&5.47&18.15e-2&3.26\\
\hline
\end{tabular}
\caption{RSE of matrix and tensor low rank decomposition relative to the best rank$-15$ space for finding important ions and pixels in the lung data set.}\label{table:pixlung}
\end{table}

For the lung data set, reconstruction errors to the best rank $r=15$ approximation based on randomized t-CX and t-CUR decompositions are given in Table \eqref{table:pixlung} based on averages over 10 independent replicates. It can be seen that the t-CX and t-CUR match or outperform their matrix variants in terms of accuracy, while improving on computing time.

These results introduce the concept of t-CX/ t-CUR tensor factorizations to MSI, describing their utility and illustrating principled algorithmic approaches to deal with the overwhelming amount of data generated by this technology and their ability to select important and intepretable ions/pixels.

\subsection{Robust PCA in the fully and partially observed settings} \label{result:completion}

Many images exhibit an inherent low rank structure and are suitably denoised by low-rank modeling methods, such as robust PCA \cite{candes2011robust}. In this section, we assess the performance of the CUR t-NN on two popular data sets, and for the typical use cases they represent. We compare the performance of the proposed CUR t-NN algorithm to the following techniques:
 \begin{itemize}
   \item
   \textbf{EXACT NN}, the exact matrix completion  \cite{candes2009exact};
   \item \textbf{RPCA NN}, the robust matrix completion  \cite{candes2011robust};
  \item \textbf{E-TUCKER NN}, the TUCKER based tensor completion \cite{liu2013tensor};
 \item \textbf{R-TUCKER NN}, the robust TUCKER based tensor completion \cite{huang2014provable};
   \item \textbf{EXACT t-NN}, the t-SVD based tensor completion \cite{zhang2017exact};
   \item \textbf{RPCA t-NN}, the robust t-SVD based tensor completion \cite{lutensor};
 \end{itemize}

All  algorithms are terminated either when the relative square error (RSE),
$$
\text{RSE}:=\frac{\|\T{L}^*- \tilde{\T{L}}\|_F}{\|\T{L}^*\|_F} \le 10^{-3},
$$
or the number of iterations and CPU times exceed 1,000 and 20 minutes, respectively.

\subsubsection{CUR t-NN on the Extended Yale B data set}\label{sec:yale}

We apply the {CUR t-NN} on the Extended Yale B data set to evaluate the accuracy of the proposed low-rank representations, as well as its computation time. The database consists of 2432 images of 38 individuals, each under 64 different
lighting conditions \cite{georghiades2001few}. We used 30 images from each subject and kept them at full resolution. Each image comprises of $192 \times 168$ pixels on a grayscale range. The data are organized into a $ 192 \times 1140 \times 168$ tensor that exhibits low tubal rank, which is an expected feature
due to the spatial correlation within lateral slices \cite{hao2013facial}. Laplacian (salt and pepper \footnote{This noise can be caused by sharp and sudden disturbances in the image signal. It demonstrates itself as sparsely occurring white and black pixels.}) noise was introduced separately in all frontal slices of the observation tensor for 20\% of the entries.

\begin{table}[h]
	\begin{center}
	\rowcolors{1}{}{lightgray}
		\begin{tabular}{|p{6.5cm}|p{1.5cm}|p{1.5cm}|}
			\hline
\textbf{Robust Completion Approach}  & \textbf{RSE}  & \textbf{Time(s)}  \\ \hline
\texttt{RPCA NN} \cite{candes2011robust} & 0.0056 &417 \\ \hline
\texttt{R-TUCKER NN} \cite{huang2014provable} & 0.0034 & 513 \\ \hline
\texttt{RPCA t-NN} \cite{lutensor}& 0.0021& 495 \\ \hline
\textbf{CUR t-NN}& \textbf{0.0026} & \textbf{136} \\ \hline			
\end{tabular}
\end{center}
\caption{RSE of tensor robust completion methods on Extended Yale B dataset.}\label{table:face}
\end{table}
We provide visualizations of the reconstructed first image from the first subject at the 20\% noise level in Figure~\ref{fig:robust_result}.  We compare the performance of the CUR t-NN, RPCA t-NN, R-TUCKER NN and RPCA NN for solving the robust tensor low rank approximation problem~\eqref{tlow7}. Table~\ref{table:face} shows the accuracy of the {CUR t-NN} together with that of competing algorithms for the Extended Yale B data set. As shown in Figure~\ref{fig:robust_result} and Table~\ref{table:face}, {CUR t-NN} estimates nearly the same face model as the {RPCA t-NN} requiring only a small fraction of time. On the other hand, all algorithms exhibit small RSE, but the {CUR t-NN} proves essentially as competitive as the best method {RPCA t-NN}, but achieves almost the same performance at approximately $1/4$ of the time.
\begin{figure}[h]
	\centering \makebox[0in]{
		\begin{tabular}{c c}
			\includegraphics[scale=0.14]{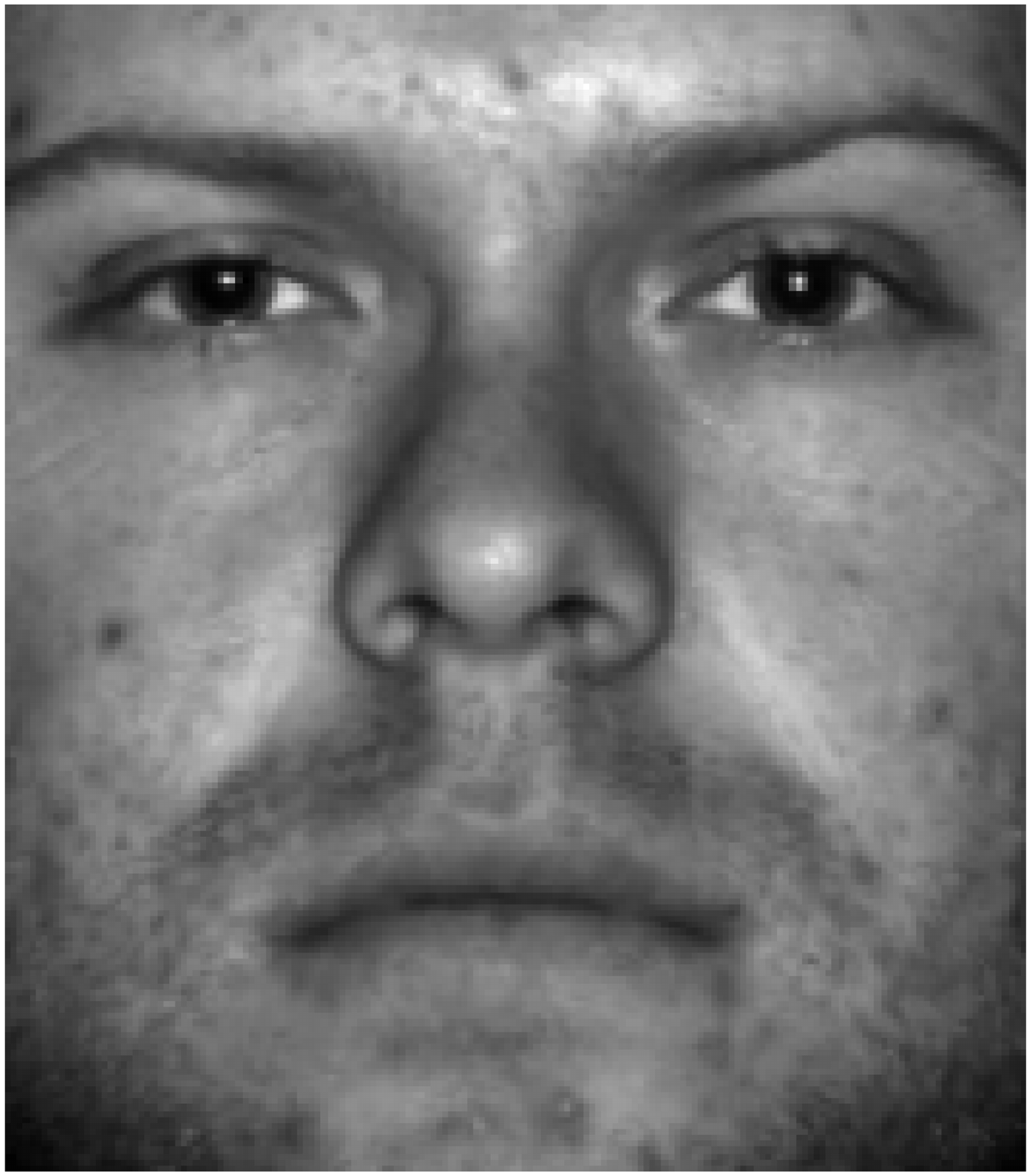}
			\includegraphics[scale=0.135]{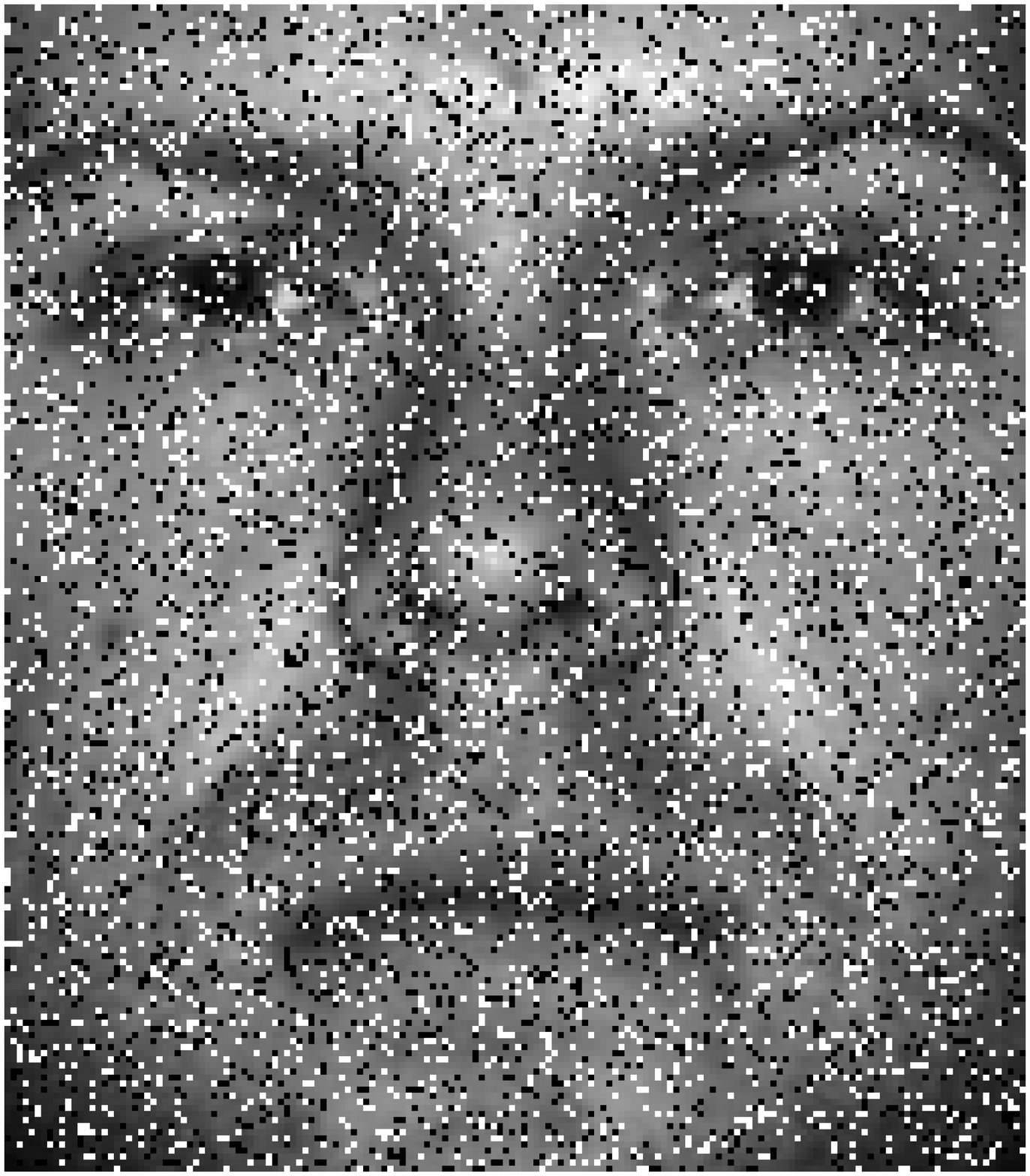} 
			\\
			(a)\\
			\includegraphics[scale=0.14]{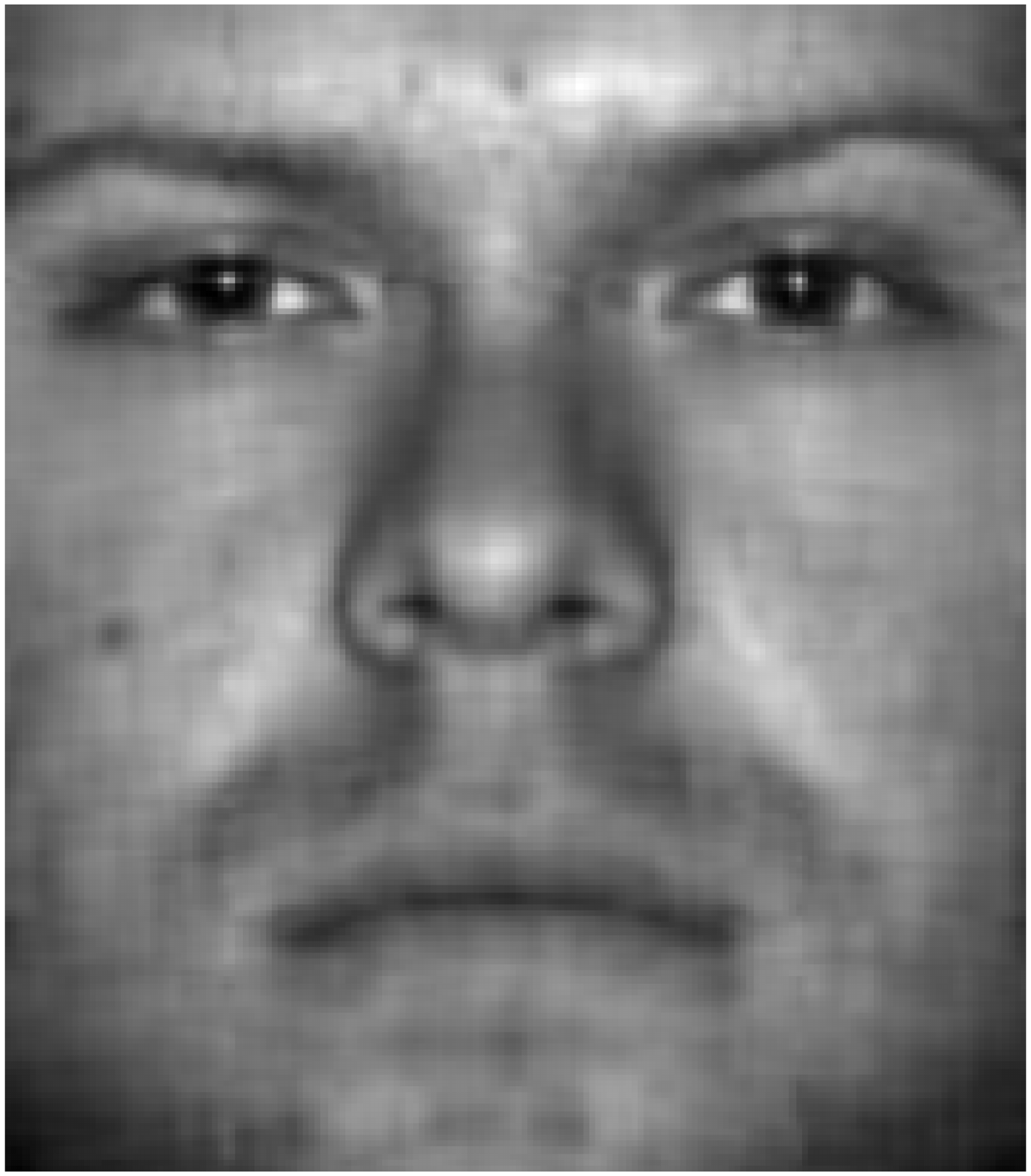}
			\includegraphics[scale=0.14]{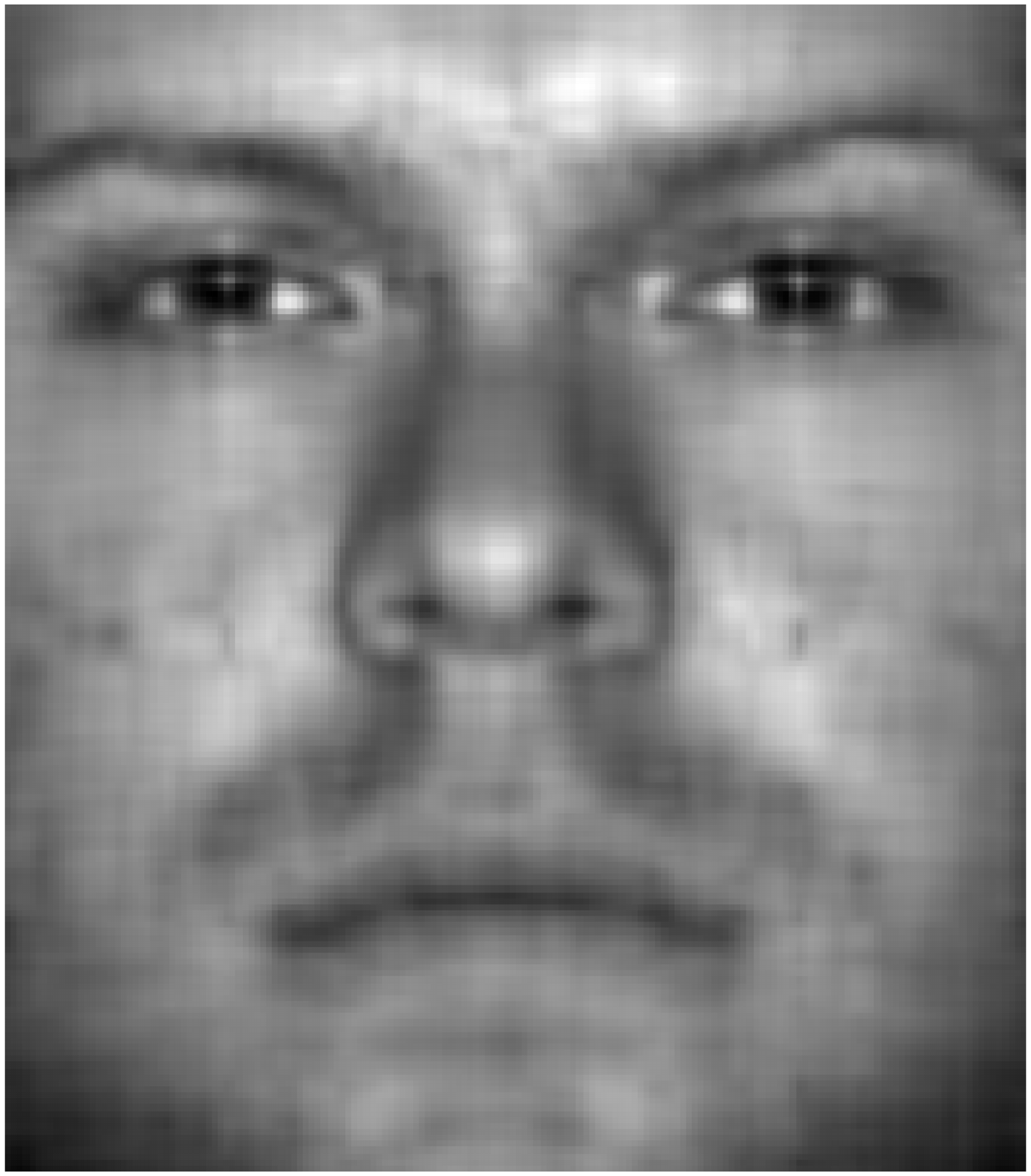} 
		\\
			(b)\\
		\end{tabular}}
\caption{ The $1$st frame of the noisy tensor factorization result for the Extended Yale B data set. \textbf{(a) Left}: The original frame. \textbf{(a)Right}: Noisy image (20\% pixels corrupted). \textbf{(b) Left}: RPCA t-NN \cite{lutensor}. \textbf{(b) Right}: CUR t-NN.}
		\label{fig:robust_result}
\end{figure} 

Beyond just speed-ups and/or accuracy improvements of the CUR t- NN algorithm, its output can be directly used in place of the singular slices and tubes that standard methods provide.  The latter represent linear combinations of the slices of the tensor, which for an image data set capture an ``average eigenface". On the other hand, CUR t-NN reconstructs the tensor
through selection of real faces in the data set, thus giving the opportunity to researchers to inspect them and examine
their representativeness. Hence, similar to the original CUR decomposition of matrices, CUR t-NN enhances the interpretability 
of the tensor decomposition \cite{hao2013facial,wright2009robust,nie2010efficient}.

\subsubsection{CUR t-NN on a video data set}

Next, we compare the {CUR t-NN} to the aforementioned listed competing methods for video data representation and compression from randomly missing entries.  The video data, henceforth referred to as the Basketball video (source: YouTube) is mapped to a  $144 \times 256 \times 80$ tensor, obtained from with a nonstationary panning camera moving from left to right horizontally following the running players. We randomly sampled
50\% entries from the Basketball video. We compare the performance of the CUR t-NN, EXACT t-NN, E-TUCKER NN and EXACT NN for solving the tensor completion problem \eqref{tlow6}.

The result is shown in Table~\ref{table:2}. It can be seen that the {CUR t-NN} outperforms almost all its competitors in terms of CPU running time and accuracy and essentially matches that of EXACT t-NN.

\begin{figure}[h]
	\centering \makebox[0in]{
		\begin{tabular}{c c}
			\includegraphics[scale=0.3]{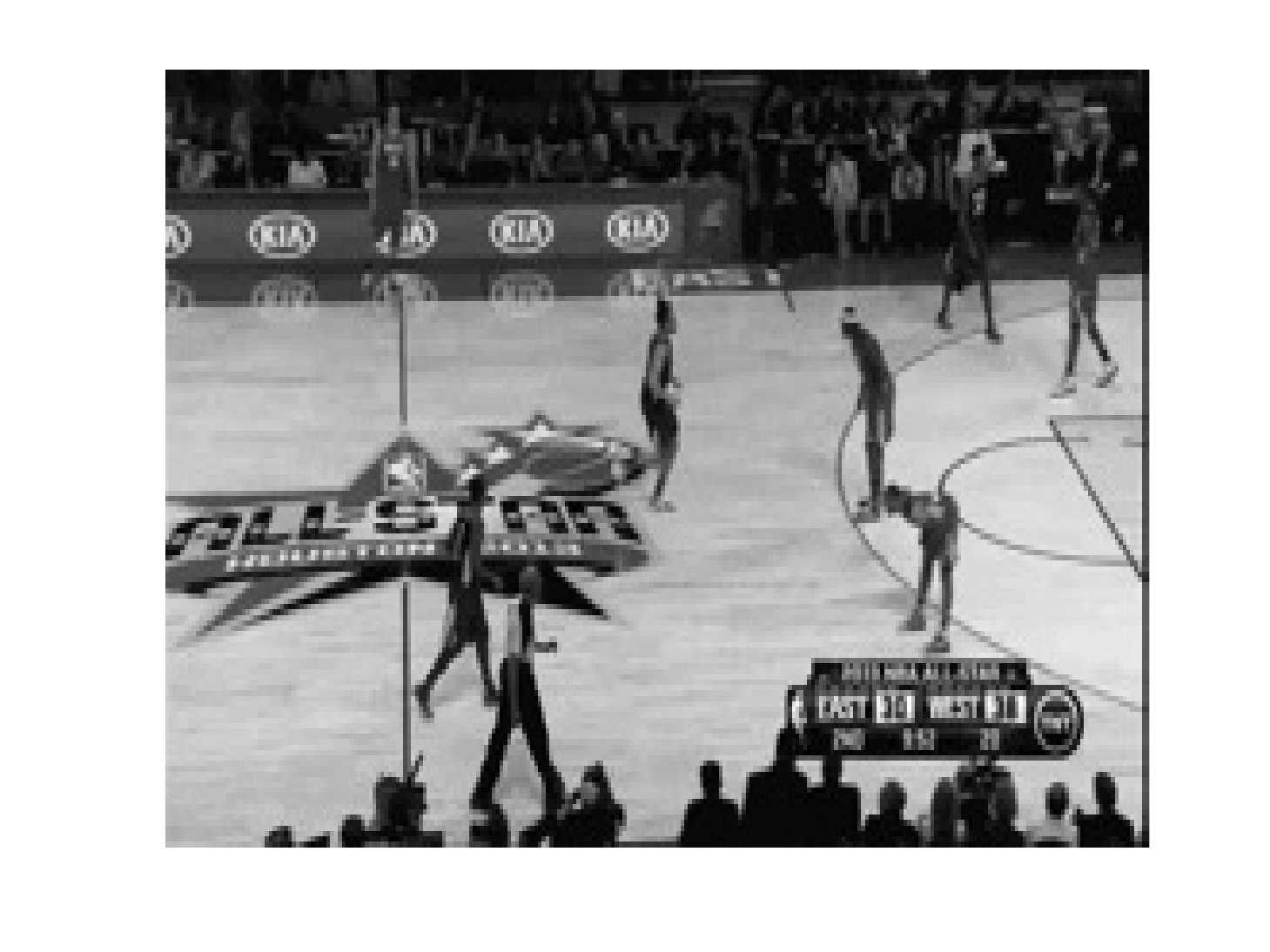}
			\includegraphics[scale=0.3]{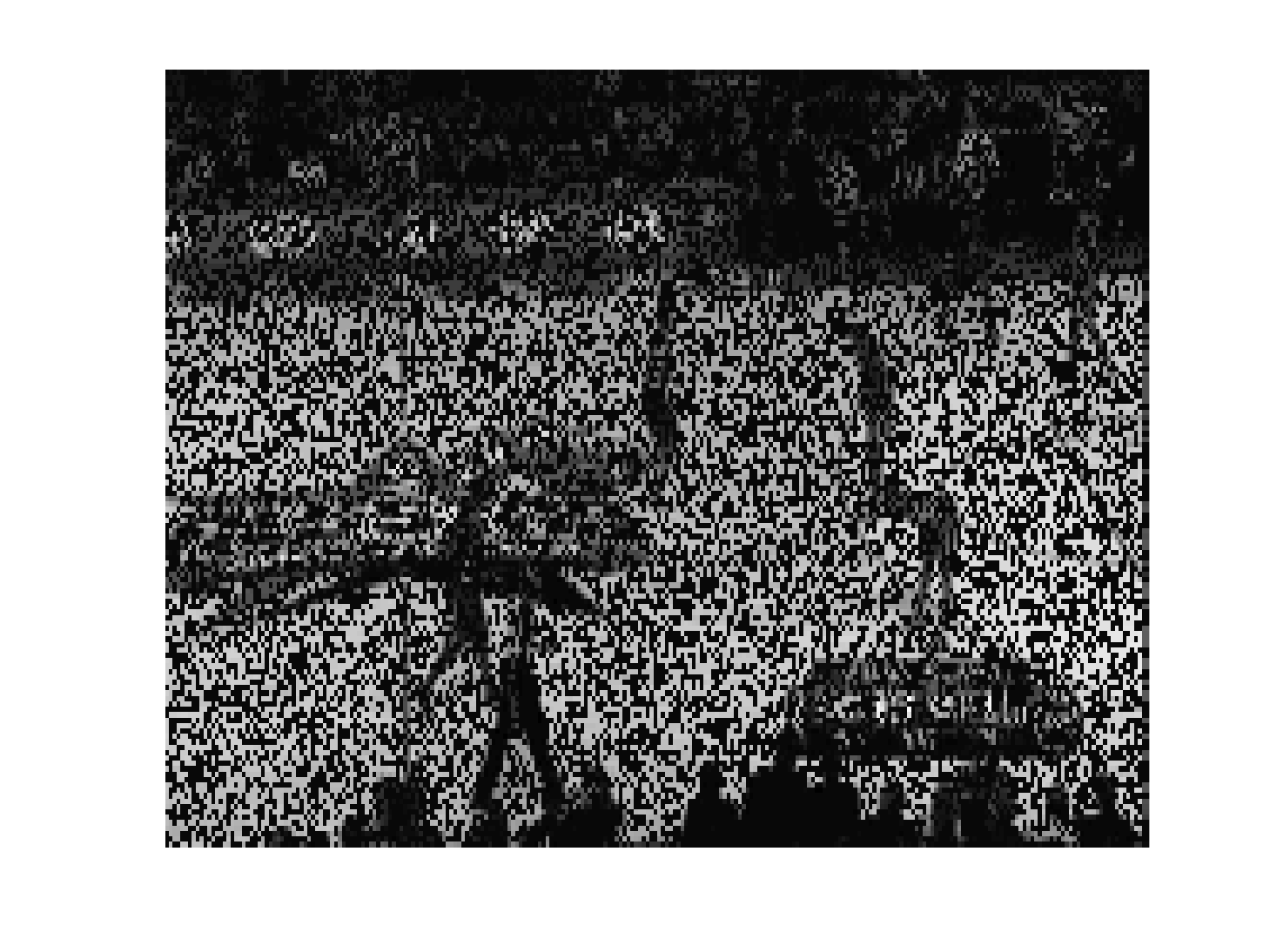} \\
			(a)\\
			\includegraphics[scale=0.3]{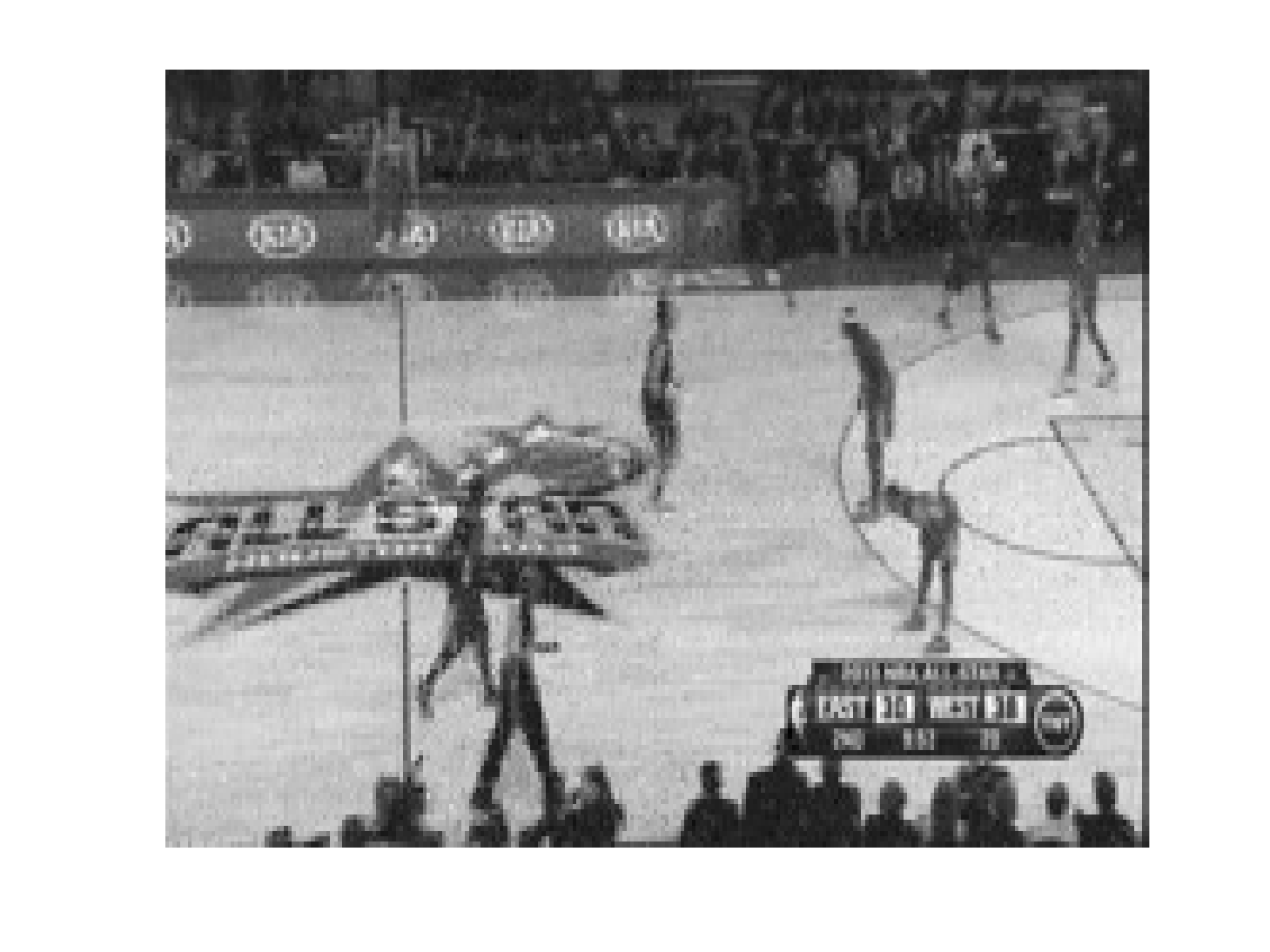}
			\includegraphics[scale=0.3]{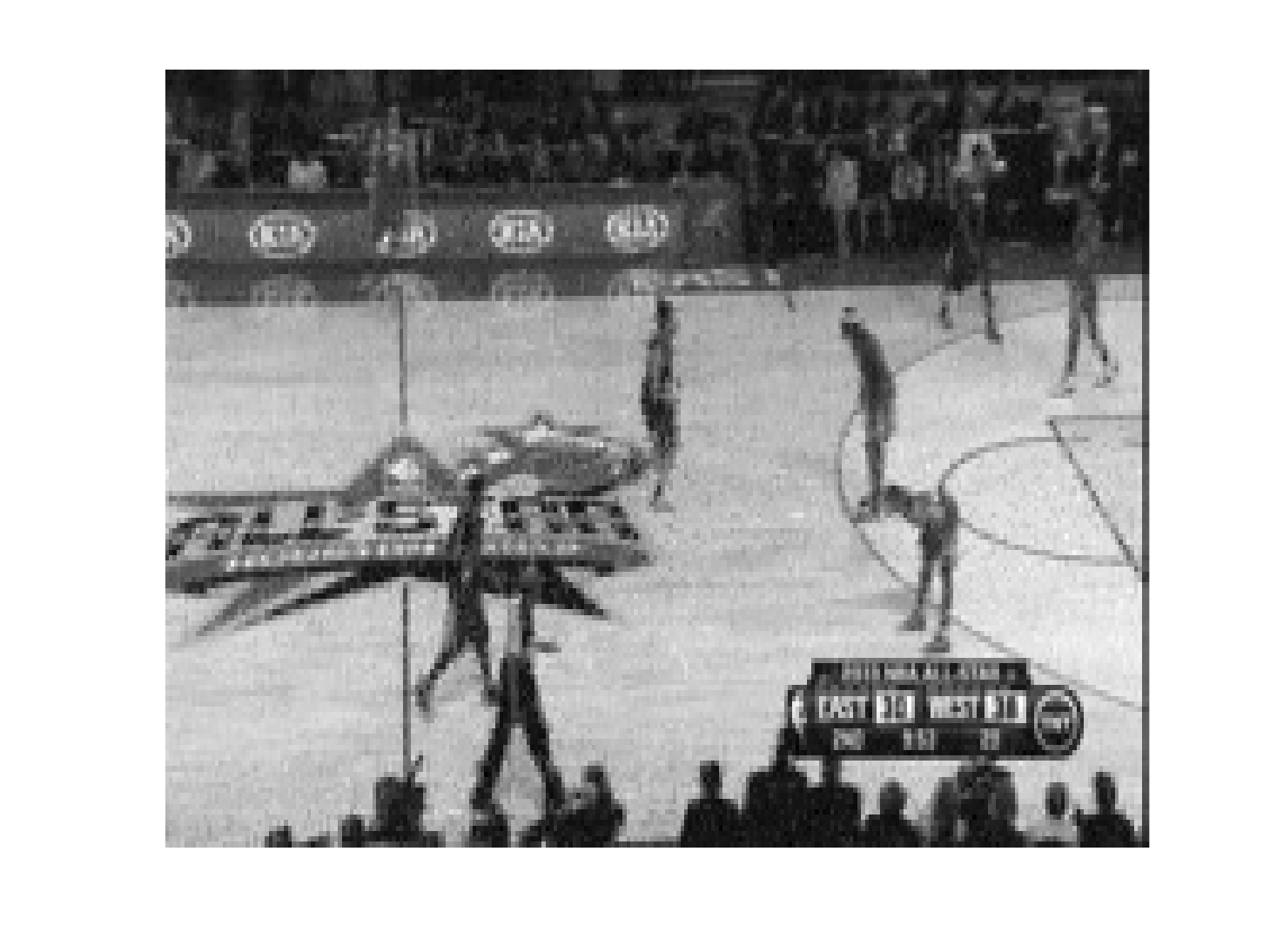} \\
			(b)\\
		\end{tabular}}
\caption{ The $20$th frame of the tensor completion result for the Basketball video. \textbf{(a) Left}: The original video. \textbf{(a) Right}: Sampled video (50\% sampling rate).  \textbf{(b) Left}: {EXACT t-NN} based reconstruction \cite{zhang2017exact}. \textbf{(b) Right}: {CUR t-NN} based reconstruction.}\label{fig:completion_result}
\end{figure}
	
\begin{table}[h]
	\begin{center}
	\rowcolors{1}{}{lightgray}
		\begin{tabular}{|p{6.5cm}|p{1.5cm}|p{1.5cm}|}
			\hline
\textbf{Completion Approach}  & \textbf{RSE}  & \textbf{Time(s)}  \\ \hline
\texttt{EXACT NN} \cite{candes2009exact} & 0.1001 & 687 \\ \hline
\texttt{E-TUCKER NN} \cite{liu2013tensor} & 0.0900 & 718 \\ \hline
\texttt{EXACT t-NN} \cite{zhang2017exact}& 0.0715 & 695 \\ \hline
\textbf{CUR t-NN} & \textbf{0.0850} & \textbf{205} \\ \hline			
\end{tabular}
\end{center}
\caption{RSE of tensor completion results for the Basketball video.}\label{table:2}
\end{table}
\section{Conclusion}

This paper introduced two randomized algorithms for basic tensor operations -rt-product and rt-project- and then used in tensor CX and CUR decomposition algorithms, whose aim is to select informative slices. The randomized tensor operations together with the tensor decompositions algorithms comes with relative error recovery guarantees. These algorithms
can be effectively used in the analysis of large scale tensors with small tubal rank.

In addition, we proposed the {CUR t-NN} algorithm that exploits the advantages of randomization for dimensionality reduction and can be used effectively for large scale noisy tensor decompositions. Indeed, CUR t-NN uses an adaptive technique to sample slices of the tensor based on a \emph{leverage score} for them and subsequently solves a convex optimization problem followed by a projection step to recover a low rank approximation to the underlying tensor. The proposed algorithm has linear running time, and provably maintains the recovery guarantees of the exact algorithm with full data tensor.

\appendix
\subsection{Proof of Theorem~\ref{theorem:matmult-main-exact}}
\label{thm:matmult-main-exact-proof}
Let $F_{n_3}$ denote the Discrete Fourier Matrix. Then, Definition~\ref{circ}  implies that
\begin{eqnarray}\label{eq:blocknorm}
\nonumber
(F_{n_3}\otimes I) \text{circ}(\T{A})(F^{-1}_{n_3}\otimes I)(F_{n_3}\otimes I) \text{unfold}(\T{A}^\top)&=&
    \begin{pmatrix}
       \hat{\T{A}}_{::1}\hat{\T{A}}^\top_{::1}
       \\&\hat{\T{A}}_{::2}\hat{\T{A}}_{::2}^\top\\\
        &&\ddots\\
        &&&\hat{\T{A}}_{::n_3}\hat{\T{A}}_{::n_3}\\\
    \end{pmatrix},\\
&=&\bdiag_{k \in [n_3]}(\hat{\T{A}}_{::k}\hat{\T{A}}^\top_{::k}).
\end{eqnarray}

Using \eqref{eq:blocknorm} and the unitary invariance of the Frobenius norm, we have 
\begin{eqnarray} \label{eqn:lab1}
\nonumber
\|\T{A}*\T{B} - \T{C}*\T{R}\|_F^2 &=&
\|(F_{n_3}\otimes I)  \text{circ}(\T{A})(F^{-1}_{n_3}\otimes I)(F_{n_3}\otimes I) \text{unfold}(\T{B}) \\
\nonumber
&-&
(F_{n_3}\otimes I) \text{circ}(\T{C})(F^{-1}_{n_3}\otimes I)(F_{n_3}\otimes I) \text{unfold}(\T{R})
\|_F^2\\
\nonumber
&=&\frac{1}{n_3} \|\bdiag_{k \in n_3} (\hat{\T{A}}_{::k}\hat{\T{B}}_{::k}- \hat{\T{C}}_{::k}\hat{\T{R}}_{::k})\|_F^2,\\
\nonumber
&=&\frac{1}{n_3} \sum_{k=1}^{n_3}\|\hat{\T{A}}_{::k}\hat{\T{B}}_{::k}- \hat{\T{C}}_{::k}\hat{\T{R}}_{::k}\|_F^2\\
&=&\frac{1}{n_3} \sum_{k=1}^{n_3}
\|\hat{\T{A}}_{::k} \hat{\T{B}}_{::k}- \hat{\T{A}}_{::k} \hat{\T{S}}_{::k} \hat{\T{D}}_{::k}\hat{\T{D}}_{::k} \hat{\T{S}}^\top_{::k} \hat{\T{B}}_{::k}\|_F^2.
\end{eqnarray}

In Algorithm~\ref{alg:matrix multiply}, we define $I_j$, $j \in [n_2]$ as an indicator variable, which is set to $1$ if both $the j$-th lateral slice of $\T{A}$ and the $j$-th horizontal slice of $\T{B}$ are selected. In this case, the selected lateral and horizontal slices are scaled by score $1/\sqrt{\min\{1,cp_j\}}$.  Note that if $\min\{1,cp_j\}=1$, then $I_j = 1$ with probability 1, and $1-I_j/\min\{1,cp_j\}=0$. Then, taking expectation on both sides of \eqref{eqn:lab1} and considering the set
$\Upsilon = \{j \in [n_2]: cp_j < 1\} \subseteq [n_2]$, we get
\begin{eqnarray}
\Ec\left[\|\T{A}*\T{B}- \T{C}*\T{R}\|_F^2 \right]
\nonumber &=&\frac{1}{n_3} \sum_{k=1}^{n_3} \Ec\left[\|{\sum_{j \in \Upsilon}
                      \left(1-\frac{I_j}{cp_j}\right) \hat{\T{A}}_{:jk} \hat{\T{B}}_{j:k}}\|_F^2\right]   \\
\nonumber &=& \frac{1}{n_3} \sum_{k=1}^{n_3}
 \Ec\left[\sum_{i_1=1}^{n_1} \sum_{i_2=1}^{n_2} \left(\sum_{j \in \Upsilon}                      \left(1-\frac{I_j}{cp_j}\right)
 \hat{\T{A}}_{:jk}\hat{\T{B}}_{j:k}\right)_{i_1 i_2}^2 \right]  \\
\nonumber &=&  \frac{1}{n_3} \sum_{k=1}^{n_3} \Ec\left[\sum_{i_1=1}^{n_1} \sum_{i_2=1}^{n_2} \left(\sum_{j \in \Upsilon}
                      \left(1-\frac{I_j}{cp_j}\right)
                      \hat{\T{A}}_{i_1jk} \hat{\T{B}}_{ji_2k}\right)^2 \right]\\
                      \nonumber
                      &=&  \frac{1}{n_3} \sum_{k=1}^{n_3} \Ec\left[\sum_{i_1=1}^{n_1} \sum_{i_2=1}^{n_2}
                      \sum_{j_1 \in \Upsilon}\sum_{j_2 \in \Upsilon}
\hat{p}_{i_1 i_2 j_1 j_2}\right]\\
          &=& \frac{1}{n_3} \sum_{k=1}^{n_3} \sum_{i_1=1}^{n_1}\sum_{i_2=1}^{n_2}
              \sum_{j_1 \in \Upsilon}\sum_{j_2 \in \Upsilon}
              \Ec \left[ \hat{p}_{i_1 i_2 j_1 j_2}\right],
\end{eqnarray}
where $\hat{p}_{i_1 i_2 j_1 j_2}=  \left(1-\frac{I_{j_1}}{cp_{j_1}}\right)\left(1-\frac{I_{j_2}}{cp_{j_2}}\right)\hat{\T{A}}_{i_1j_1k}\hat{\T{B}}_{j_1i_2k} \hat{\T{A}}_{i_1j_2k}\hat{\T{B}}_{j_2i_2k}$.
%

Since for $j \in [\Upsilon]$, $\Ec[1-I_{j}/cp_{j}]=0$ and $\Ec[\left(1-I_{j}/cp_{j}\right)^2]=\left(1/cp_j\right)-1 \leq 1/cp_j$, we get
\begin{eqnarray}\label{ex:side}
\nonumber
\Ec\left[\|\T{A}*\T{B}- \T{C}*\T{R}\|_F^2 \right]
&=&\frac{1}{n_3} \sum_{k=1}^{n_3}  \sum_{i_1=1}^{n_1} \sum_{i_2=1}^{n_2} \sum_{j \in \Upsilon}
          \Ec{\left(1-I_{j}/cp_{j}\right)^2} \hat{\T{A}}_{i_1jk}^2 \hat{\T{B}}_{ji_2k}^2   \\
\nonumber
&\leq&\frac{1}{n_3} \sum_{k=1}^{n_3} \sum_{j \in \Upsilon} \frac{1}{cp_j}
          \sum_{i_1=1}^{n_1} \sum_{i_2=1}^{n_2} \hat{\T{A}}_{i_1jk}^2 \hat{\T{B}}_{ji_2k}^2\\
\nonumber
&=& \frac{1}{n_3} \sum_{k=1}^{n_3} \frac{1}{c} \sum_{j \in \Upsilon}
     \frac{\|{\hat{\T{A}}_{:jk}}\|_2^2{\|\hat{\T{B}}_{j:k}\|_2^2}}{p_j}\\
     \nonumber
     &=& \frac{1}{c} \sum_{j \in \Upsilon}
     \frac{\|\T{A}_{:j1}\|_2^2 \|\T{B}_{j:1}\|_2^2}{p_j} + \dots + \frac{1}{c} \sum_{j \in \Upsilon}
     \frac{\|\T{A}_{:jn_3}\|_2^2\|\T{B}_{j:n_3}\|_2^2}{p_j}\\
&=& \frac{1}{c} \sum_{j \in \Upsilon}
     \frac{\|\T{A}_{:j:}\|_F^2\|\T{B}_{j::}\|_F^2}{p_j}.
\end{eqnarray}

Equation (\ref{eqn1:theorem:matmult-main-expected}) follows from (\ref{ex:side}) by using Jensen's inequality and the fact that the sampling probabilities (\ref{eqn:nearly_optimal_probs})
and (\ref{eqn:nonoptimal_probs}) are defined in the original domain.
\subsection{Proof of Lemma~\ref{lem:mat-mult}}
\label{lem:lem:mat-mult-proof}

Using Definitions~\ref{circ} and \ref{def:spec:ten}, we have that
\begin{eqnarray} \label{eq:normmax:mult}
\nonumber
\|\T{A}*\T{A}^\top - \T{C}*\T{C}\| &=&
\|(F_{n_3}\otimes I)  \text{circ}(\T{A})(F^{-1}_{n_3}\otimes I)(F_{n_3}\otimes I) \text{unfold}(\T{A}^\top) \\
\nonumber
&-&
(F_{n_3}\otimes I) \text{circ}(\T{C})(F^{-1}_{n_3}\otimes I)(F_{n_3}\otimes I) \text{unfold}(\T{C}^\top)
\|_2,\\
\nonumber
&=&\|\bdiag_{k \in n_3}(\hat{\T{A}}_{::k}\hat{\T{A}}_{::k}^\top- \hat{\T{C}}_{::k}\hat{\T{C}}_{::k}^\top)\|_2,\\
&=& \max_{k \in [n_3]} \big\{\| \hat{\T{A}}_{::k}\hat{\T{A}}_{::k}^\top- \hat{\T{C}}_{::k}\hat{\T{C}}_{::k}^\top\|_2\big\},
\end{eqnarray}
where the first equality follows from the unitary invariance of the spectral norm, the second equality from \eqref{eq:blocknorm} and the third equality follows since the spectral norm of a block matrix is equal to the maximum of block norms.

Let $k_\pi =k \in [n_3]$ be the index of the tensor's frontal slice with maximum spectral norm $\| \hat{\T{A}}_{::k}\hat{\T{A}}_{::k}^\top- \hat{\T{C}}_{::k}\hat{\T{C}}_{::k}^\top\|_2$. Using \eqref{eq:normmax:mult},  and Example~4.3 in \cite{hsu2012tail}, we obtain
\begin{eqnarray} \label{eq:normmatrix:mult}
\| \hat{\T{A}}_{::k_\pi}\hat{\T{A}}_{::k_\pi}^\top- \hat{\T{C}}_{::k_\pi}\hat{\T{C}}_{::k_\pi}^\top\|_2 \leq \frac{\epsilon}{2} \|\hat{\T{A}}_{::,k_\pi}\|_2^2,
\end{eqnarray}
 with probability at least $1-\delta$.
\subsection{Proof of Lemma~\ref{lem:eps}}
\label{lem:eps-proof}


Let $\sigma_{i,k}$ be the $i$-th largest singular value of slice $k$. For all $1\leq i \leq r$ and $1 \leq k \leq n_3$, we have
\begin{eqnarray}\label{sing:bound}
\nonumber
|1-\max_{k \in [n_3]} \sigma_{i,k}^2(\T{V}^\top*\T{S} *\T{D})| &=&
|\max_{k \in [n_3]} \sigma_{i,k}(\T{V}^\top*
\T{V})- \max_{k \in [n_3]} \sigma_{i,k}(\T{V}^\top*\T{S}*\T{D}*\T{D}*\T{S}^\top * \T{V})| \\
&\leq& 
\|\T{V}*\T{V}^\top
-\T{V}^\top * \T{S}*\T{D}* \T{D}*\T{S}^\top * \T{V}\|_2.
\end{eqnarray}
 

Using Lemma~\ref{lem:mat-mult} and \eqref{sing:bound}, we get
\begin{equation}\label{singular:lowbound}
|1-\max_{k \in [n_3]} \sigma_{i,k}^2(\T{V}^\top*\T{S} *\T{D})| \leq \epsilon/2 \|\hat{\T{V}}_{::k_\pi}\|_2^2
\end{equation}
for all $1\leq i \leq r$ and an index $k_\pi \in [n_3]$.

Since $\epsilon \in(0,1]$, each tubal singular value of $\T{V}^\top *\T{S}$ is positive, which implies that $\rank_t ({\T{V}^\top *\T{S}}) = \rank_t({\T{V}}) = \rank_t({\T{L}})$.

To prove \eqref{lem:eps2}, we use the t-SVD of $\T{V}^\top * \T{S}*\T{D}$ and note that
\begin{eqnarray}\label{bound:omega}
\nonumber
\|{\Omega}\|
  &=& \|{\left(\T{V}^\top*\T{S}* \T{D}\right)^\dagger - \left(\T{V}^\top*\T{S}* \T{D}\right)^\top}\|\\
  \nonumber
  &=& \|{\left(\T{U}_{\T{V}^\top*\T{S}* \T{D}}\Sigma_{\T{V}^\top*\T{S}* \T{D}}\T{V}_{\T{V}^\top*\T{S}* \T{D}}^\top\right)^\dagger - \left(\T{U}_{\T{V}^\top*\T{S}* \T{D}}\Sigma_{\T{V}^\top*\T{S}* \T{D}}\T{V}_{\T{V}^\top*\T{S}* \T{D}}^\top\right)^\top}\|\\
  \nonumber
  &=& \|{\T{V}_{\T{V}^\top*\T{S}* \T{D}}\left(\Sigma_{\T{V}^\top*\T{S}* \T{D}}^{-1} - \Sigma_{\T{V}^\top*\T{S}* \T{D}}\right)\T{U}_{\T{V}^\top*\T{S}* \T{D}}^\top}\|\\
    &=& \|\left(\Sigma_{\T{V}^\top*\T{S}* \T{D}}^{-1} - \Sigma_{\T{V}^\top*\T{S}* \T{D}}\right)\|.
\end{eqnarray}
The claim follows since $\T{V}_{\T{V}^\top*\T{S}* \T{D}}$ and $\T{U}_{\T{V}^\top*\T{S}* \T{D}}$ are orthogonal tensors.

To prove \eqref{lem:eps3}, note that
\begin{eqnarray}
\nonumber
\left(\T{L}*\T{S}*\T{D}\right)^\dagger
   &=& \left(\T{U}*\Sigma*\T{V}^\top*\T{S}*\T{D}\right)^\dagger          \\
\nonumber
   &=& \left(\T{U}*\Sigma* \T{U}_{\T{V}^\top*\T{S}* \T{D}}*\Sigma_{\T{V}^\top*\T{S}* \T{D}}*\T{V}_{\T{V}^\top*\T{S}* \T{D}}^\top\right)^\dagger \\
\label{eq20}
   &=& \T{V}_{\T{V}^\top*\T{S}* \T{D}}* \left(\Sigma* \T{U}_{\T{V}^\top*\T{S}* \T{D}}*\Sigma_{\T{V}^\top*\T{S}* \T{D}}\right)^\dagger *\T{U}^\top.
\end{eqnarray}
To remove the pseudoinverse in the above derivations, we use the first part of the Lemma. In this case,
\begin{eqnarray}
\nonumber
\left(\Sigma* \T{U}_{\T{V}^\top*\T{S}* \T{D}}*\Sigma_{\T{V}^\top*\T{S}* \T{D}}\right)^\dagger
   &=& \left(\Sigma* \T{U}_{\T{V}^\top*\T{S}* \T{D}}*\Sigma_{\T{V}^\top*\T{S}* \T{D}}\right)^{-1} \\
\label{eq201}
   &=& \Sigma_{\T{V}^\top*\T{S}* \T{D}}^{-1}*\T{U}_{\T{V}^\top*\T{S}* \T{D}}^\top \Sigma^{-1}.
\end{eqnarray}
By combining (\ref{eq20}) and (\ref{eq201}), we get the result.

To prove \eqref{lem:eps4}, we have that for all $1\leq i \leq r$ and $1 \leq k \leq n_3$,
   \begin{align*}
    \|\Omega \| & = \|\Sigma_{\T{V}^\top*\T{S}* \T{D}}^{-1} - \Sigma_{\T{V}^\top*\T{S}* \T{D}} \| & \textnormal{from \eqref{bound:omega}}\\
    & = \max_{i,k}  \left|  \sigma_{i,k}(\T{V}^\top*\T{S}* \T{D}) -  \frac{1}{\sigma_{i,k}(\T{V}^\top*\T{S}* \T{D})} \right| & \textnormal{by definition} \\
    & = \max_{i,k} \frac{| \sigma_{i,k}^2(\T{V}^\top*\T{S}* \T{D}) -  1 |}{|\sigma_{i,k}(\T{V}^\top*\T{S}* \T{D})|} & \textnormal{by simple manipulation} \\
    & \leq \frac{\epsilon/2}{\sqrt{1 - \epsilon/2}} & \textnormal{from \eqref{sing:bound}}\\
    & \leq  \epsilon/\sqrt{2} & \textnormal{if } \epsilon < 1 =>\sqrt{1 - \epsilon/2} > 1/\sqrt{2}.
    \end{align*}

\subsection{Proof of Theorem~\ref{thm:regress-main}}\label{proofregressmain}

From Definition~\ref{eq:ten:incoh},  and using \eqref{eq:cur-col-prob} with $\beta=n_3/\mu_0(\T{V})\in (0,1]$, we have that
\begin{eqnarray*}
\frac{\beta}{r n_3}\|\hat{\T{V}}_{i::}\|_F^2 = \frac{\beta}{r }\|\T{V}_{i::}\|_F^2 \leq \frac{\beta}{r}\frac{r}{n_2 n_3}\mu_0(\T{V})
   = \frac{1}{n_2} = p_i,
\end{eqnarray*}
for all $i \in [n_2]$.

Using Proposition \ref{prop:regress}, our choice of $c$ implies that
\begin{eqnarray} \label{eq:prop:regress}
\nonumber 
  \|\T{A} - \T{A}_c* \T{L}_c^\dagger*\T{L}\|_F &=&\|\T{A} -
\T{A}*{\T{C}*\T{D}}*(\T{L}_c*\T{D})^\dagger*\T{L}\|_F,\\
&\leq& (1+\epsilon)\|\T{A}-\T{A}*\T{L}^\dagger*\T{L}\|_F.
\end{eqnarray}
holds with probability at least $0.85$. 

\subsection{Proof of Corollary~\ref{cor:proj-main}}
\label{cor:proj-main-proof}

Since $\T{C}^+*\T{A}$ minimizes $\|\T{A}-\T{C}*\T{X}\|_F$ over all tensors $\T{X} \in \mathbb{R}^{n_1 \times c \times n_3}$, it follows that $$\|\T{A}-\T{C}*\T{C}^\dagger*\T{A}\|_F\leq\|\T{A}-\T{C}*\T{L}_{\T{C}}^\dagger*\T{L}\|_F.$$

Now, using  \eqref{eq:prop:regress}, we get 
\begin{equation}\label{eq:tcx:constant}
 \|\T{A} - \T{C}*\T{C}^\dagger*\T{A}\|_F \leq (1+\epsilon)\|\T{A} - \T{L}\|_F,
\end{equation} 
 with probability at least $0.85$.

%
We can trivially boost the success probability to $1-\delta$ by repeating Algorithm~\ref{alg:algCX} $O (\log (1/\delta))$ rounds. Specifically, let $\T{C}_i$ denote the output of Algorithm~\ref{alg:algCX} at round $i$;  using \eqref{eq:tcx:constant} for each $\T{C}_i $ we have
\begin{equation}
\label{eqn:rela}
\|\T{A}  - \T{C}_i * \T{C}_i^\dagger*\T{A}\|_F\leq (1+\epsilon)\|\T{A} - \T{L}\|_F,  
\end{equation}
with probability at least $0.85$.

Now, let $\T{C}$ denote the set of all columns used in the approximation. Since  each $\T{C}_i = \T{C} * \T{S}_i$ for some tensor $\T{S}_i$ and $\T{C}^+*\T{A}$ minimizes $\|\T{A}-\T{C}*\T{X}\|_F$ over all tensors $\T{X} \in \mathbb{R}^{n_1 \times c \times n_3}$, it follows that $$\|\T{A}-\T{C}*\T{C}^\dagger*\T{A}\|_F\leq\|\T{A}-\T{C}_i*\T{C}_i^\dagger*\T{A}\|_F,$$ 
for each $i$.  
Hence, if 
$$ \|\T{A} - \T{C}*\T{C}^\dagger*\T{A}\|_F \leq (1+\epsilon)\|\T{A} - \T{L}\|_F,$$
fails to hold, then for each $i$, \eqref{eqn:rela} also fails to hold.
Since $0.15< 1/e$, the desired conclusion must hold with probability at least $1- (1/e)^{\log (1/\delta)}  = 1-\delta$.
%
%
%

\subsection{Proof of Corollary~\lowercase{\ref{cor:gnys-main}}}
\label{sec:gnys-proof}

Using \eqref{eq:prop:regress}, it follows that
\begin{equation} \label{eq:tproj}
\|\T{A} - \T{C}*\T{U}*\T{R}\|_F \leq (1+\epsilon)\|\T{A} - \T{C}*\T{C}^\dagger*\T{A}\|_F,
\end{equation}
with probability at least $0.85$.

Using \eqref{eq:tcx:constant} and \eqref{eq:tproj}, our choice of $c$ and $l$  implies the following holds with probabilities at least $0.7$, 
\begin{equation}\label{eq:tcurproj}
\|\T{A} - \T{C}*\T{U}*\T{R}\|_F \leq (1+\epsilon)^2|\T{A} - \T{L}\|_F \leq (1+\epsilon')|\T{A} - \T{L}\|_F,
\end{equation}
where $\epsilon'= 3\epsilon$. 

The inequality \eqref{eq:tcurproj} holds with probability at least $1-\delta$ by following the boosting strategy employed in the proof of Corollary~\ref{cor:proj-main}.  Indeed, since in each trial inequality \eqref{eq:tcurproj} fails with probability less than $0.3 < 1/e$, the claim of Corollary~\lowercase{\ref{cor:gnys-main}} will hold with probability greater than $1 -(1/e)^{\log (1/\delta)}= 1 -\delta$. 
\subsection{Proof of Lemma~\lowercase{\ref{lem:sub-coh}}}

Since from Lemma~\ref{lem:eps} we have $\rank_t(\T{L}_c)=\rank_t(\T{L})$, the first claim follows similarly by using Lemma~1 of~\cite{mohri2011can}.

To prove the second claim, using Lemma~\ref{lem:eps}, assume that $\T{S}^\top*\T{V}$ consists of the first $c$ horizontal slices of $\T{V}$.  Then, if 
$\T{L}_c = \T{U}*\Sigma* \T{V}^\top*\T{S}$ has $\rank_t(\T{L}_c)=\rank_t(\T{L})=r$, the tensor $\T{V}^\top*\T{S}$ must have full tubal rank. Thus, we can write
\begin{eqnarray*}
\T{L}_c^+*\T{L}_c&=& (\T{U}*\Sigma*\T{V}^\top*\T{S})^+ *\T{U}*\Sigma*\T{V}^\top*\T{S}\\
&=& (\Sigma*\T{V}^\top*\T{S})^+*\T{U}^\top*\T{U}*\Sigma*\T{V}^\top*\T{S}\\
&=&(\Sigma*\T{V}^\top*\T{S})^+*\Sigma*\T{V}^\top*\T{S}\\
&=&(\T{V}^\top*\T{S})^+*\Sigma^\top*\Sigma*\T{V}^\top*\T{S}\\
&=&(\T{V}^\top*\T{S})^+*\T{V}^\top*\T{S}\\
&=&\T{V}^\top*\T{S} *(\T{V}^\top*\T{S}*\T{S}^\top*\T{V})^{-1}\T{V}^\top*\T{S},
\end{eqnarray*}%
where the second and third equalities follow from $\T{U}^\top*\T{U}=\T{I}_c$. The fourth and fifth equalities result from $\Sigma$ having full tubal rank and $\T{V}$ having full lateral slice rank, and the sixth follows from $\T{S}^\top*\T{V}$ having
full horizontal slice rank. Next, denote the right singular vectors of $\T{L}_c$ by $\T{V}_c \in \mathbb{R}^{c\times r \times n_3}$.
Define $\tcol{e}_{i,c}$ as the $i$-th lateral slice of $\T{I}_c$ and $\tcol{e}_{i,n_2}$ as the $i$-th lateral slice of $\T{I}$. Then we have,
\begin{eqnarray*}
\mu_0(\T{V}_c) &=& \frac{cn_3}{r} \max_{1 \le i \le c} \|\T{V}_c^\top* \tcol{e}_{i,c}\|_F^2 \\
&=& \frac{cn_3}{r} \max_{1 \le i \le c} \trace\{ \tcol{e}_{i,c}^\top*\T{L}_c^+*\T{L}_c*\tcol{e}_{i,c} \}\\
&=& \frac{cn_3}{r} \max_{1 \le i \le c} \trace\{ \tcol{e}_{i,c}^\top*(\T{V}^\top*\T{S})^+*\T{V}^\top*\T{S}* \tcol{e}_{i,c}\}\\
&=& \frac{cn_3}{r} \max_{1 \le i \le c} \trace\{ \tcol{e}_{i,c}^\top * \T{S}^\top*\T{V}*\T{W}^{-1}*\T{V}^\top*\T{S} *\tcol{e}_{i,c}\}\\
&=& \frac{cn_3}{r} \max_{1 \le i \le c} \trace\{ \tcol{e}_{i,n_2}^\top*\T{V}*\T{W}^{-1}*\T{V}^\top *\tcol{e}_{i,n_2}\},
\end{eqnarray*}
where $\T{W} =\T{V}^\top*\T{S}* \T{S}^\top*\T{V}$ and the final equality follows from $\T{V}^\top*\T{S}* \tcol{e}_{i,c} = \T{V}^\top*\tcol{e}_{i,n_2}$ for all $1\le i\le c$.

Next, we have
\begin{eqnarray*}
\mu_0(\T{V}_{\T{L}_{\T{C}}})&=& \frac{cn_3}{r} \max_{1 \le i \le c}
\trace\{\tcol{e}_{i,n_2}^\top *\T{V}*\T{W}^{-1} *\T{V}^\top *\tcol{e}_{i,n_2}\} \\
&=& \frac{cn_3}{r} \max_{1 \le i \le c} \trace\{\T{W}^{-1}*\T{V}^\top*\tcol{e}_{i,n_2}*
\tcol{e}_{i,n_2}^\top * \T{V}\}\\
&\le& \frac{cn_3}{r} \|\T{W}^{-1}\|^2 \max_{1 \le i \le c}
\|\T{V}^\top*\tcol{e}_{i,n_2}* \tcol{e}_{i,n_2}^\top *\T{V}\|_\circledast^2 ,
\end{eqnarray*}
where the last inequality follows form H\"{o}lder's inequality.

Since $\T{V}^\top*\tcol{e}_{i,n_2} *\tcol{e}_{i,n_2}^\top*\T{V}$ has tubal rank one, using the definition of $\mu_0$-coherence, we have
$$
\mu_0(\T{V}_{\T{L}_{\T{C}}})
\leq \frac{c}{n_2} \|\T{W}^{-1}\|^2 \mu_0(\T{V}).
$$
Now, using \eqref{singular:lowbound}, we have that
$\|\T{W}^{-1}\|^2 \leq \frac{n_2}{c(1-\epsilon/2)}$. Thus, it follows that $$\mu_0(\T{V}_{\T{L}_{\T{C}}}) \leq  \mu_0(\T{V})/(1-\epsilon/2).$$

To prove the last claim under Lemma~\ref{lem:eps}, we note that
\begin{eqnarray*}
\mu_1(\T{L}_c) &=& \frac{n_1cn_3^2}{r}\max_{\substack{1\le i \le n_1\\ 1\le j \le
c}}\|\tcol{e}_{i,n_1}^\top*\T{U}_{c}*\T{V}_{c}^\top* \tcol{e}_{j,c}\|_F^2 \\
&\leq& \frac{n_1cn_3^2}{r} \max_{1\le i \le n_1} \|\T{U}_{c}^\top *\tcol{e}_{i,n_1}\|_F^2 \max_{1\le j \le c}\|\T{V}_{c}^\top *\tcol{e}_{j,c}\|_F^2 \\
&\leq& \frac{r}{(1-\epsilon/2)}\mu_0(\T{U})\mu_0(\T{V}).
\end{eqnarray*}

%
\subsection{Proof of Theorem~{\ref{thm:master}}}
\label{sec:proof-master}

Define $A(\T{X})$ as the event that a tensor $\T{X}$ is $(\frac{r\mu^2}{1-\epsilon/2},r)$-coherent. To prove Theorem~\ref{thm:master}, let $\tilde{\T{L}}$ denote the solution obtained by CUR t-NN and let $\T{L}^*$ be the exact solution of problems~\eqref{tlow6}, \eqref{tlow7} and \eqref{tlow8}. In Algorithm~\ref{CUR t-NN}, define $\bar{\T{L}}$ as
$$
\bar{\T{L}} =
\begin{bmatrix}
\tilde{\T{C}}_{1} & \tilde{\T{R}}_{2} \\
 \tilde{\T{C}}_{2} & \T{L}_{22}^*
\end{bmatrix}\,,
$$
where $ \tilde{\T{C}} =  \begin{bmatrix}
      \tilde{\T{C}}_{1} & \tilde{ \T{C}}_{2}
      \end{bmatrix}^\top$,  and $ \tilde{\T{R}} =  \begin{bmatrix}
    \tilde {\T{R}}_{1} &  \tilde{\T{R}}_{2}
     \end{bmatrix}$,  and $\T{L}_{22}^* \in \mathbb{R}^{(n_1-l) \times (n_2-c) \times n_3}$ is the bottom right subtensor of $\T{L}^*$.  It can easily be seen that 
\begin{equation}\label{curnn}
\|\T{L}^* -\bar{\T{L}}\|_F^2  \leq
\|\T{C}^*-\tilde{\T{C}}\|_F^2+\|\T{R}^*-\tilde{\T{R}}\|_F^2,
\end{equation}     
Now, define $W(\tilde{\T{L}},\bar{\T{L}})$ as the event
\begin{eqnarray}\label{base:sol}
\|\tilde{\T{L}}-\bar{\T{L}}\|_F \leq (1 + \epsilon)\|\T{L}^* -\bar{\T{L}}\|_F.
\end{eqnarray}

If  $W(\tilde{\T{L}},\bar{\T{L}})$ holds, we have  
\begin{align*}
 \|\T{L}^* - \tilde{\T{L}}\|_F &\leq  \|\T{L}^* - \bar{\T{L}}\|_F+\|\bar{\T{L}} - \tilde{\T{L}}\|_F & \textnormal{by the triangle inequality}\\
&\leq \|\T{L}^* - \bar{\T{L}}\|_F + (1+\epsilon) \|\T{L}^* - \bar{\T{L}}\|_F & \textnormal{ from \eqref{base:sol}} \\
&\leq (2 + \epsilon) \|\T{L}^* - \bar{\T{L}}\|_F & \\
& \leq (2 + \epsilon)  \sqrt{\|\T{C}^*-\tilde{\T{C}}\|_F^2+\|\T{R}^*-\tilde{\T{R}}\|_F^2} &  \textnormal{from \eqref{curnn}}.
\end{align*}

Next, we consider all three events $W(\tilde{\T{L}},\bar{\T{L}})$, $A(\T{R}^*)$, and $A(\T{C}^*)$ with $\log(3/\delta)$. Using Lemma~\ref{lem:sub-coh}, our choice of $c$ and $l$ implies that $A(\T{C}^*)$ and $A(\T{R}^*)$ holds with probability at least $1-\delta/3$. Since $\tilde{\T{L}}$ is a t-CUR approximation to $\bar{\T{L}}$, from Corollary~\ref{cor:gnys-main}, we get that $W(\tilde{\T{L}},\bar{\T{L}})$ holds with probability at least $1-\delta/3$. 

Using the union bound, it follows that
\begin{eqnarray}\label{eq:final:prob}
\nonumber 
\textbf{Pr}(W(\tilde{\T{L}},\bar{\T{L}}) \bigcap A(\T{C}^*) \bigcap A(\T{R}^*)) &\geq & 1-\textbf{Pr}(W(\tilde{\T{L}},\bar{\T{L}})^c)-\textbf{Pr}(A(\T{C}^*)^c)-\textbf{Pr}(A(\T{R}^*)^c) \\
\nonumber
&\geq& 1 - \delta/3 -\delta/3- \delta/3 \\
\nonumber 
&=& 1-\delta.
\end{eqnarray}
\subsection{Optimization by ADMM}\label{sec:admm:app}

We provide the optimization and parameter setting details of ADMM used in Algorithm \ref{alg1} for solving a robust tensor factorization.
\begin{algorithm}[!h]\caption{,~~$\T{L} \leftarrow \textbf{ADMM}(\T{X}, \lambda)$}
		\textbf{Initialize:} $\T{L}_0=\T{E}_0=\T{Y}_0=0$, $\rho=1.1$, $\mu_0=1e-3$, $\mu_{\max}=1e+10$, $\epsilon=1e-8$.   \\ 
	\textbf{while} not converged \textbf{do}
	\begin{itemize}
		\item $\T{L}_{k+1} \leftarrow \argmin_{\T{L}}\|\T{L}\|_*+\frac{\mu_k}{2}\|\T{L}+\T{E}_{k}-\T{X}+\frac{\T{Y}_{k}}{\mu_k}\|_F^2$;
		\item  $ \T{E}_{k+1} \leftarrow \argmin_{\T{E}} \lambda\|\T{E}\|_1+\frac{\mu_k}{2}\|\T{L}_{k+1}+\T{E}-\T{X}+\frac{\T{Y}_{k}}{\mu_k}\|_F^2$;
		\item $\T{Y}_{k+1}=\T{Y}_{k}+\mu_k(\T{L}_{k+1}+\T{E}_{k+1}-\T{X})$;
		\item Update $\mu_{k+1}$ by $\mu_{k+1}=\min(\rho\mu_k,\mu_{\max})$;
		\item Check $ \|\T{L}_{k+1}-\T{L}_{k}\|_\infty\leq\epsilon, \ \|\T{E}_{k+1}-\T{E}_{k}\|_\infty\leq\epsilon,  \|\T{L}_{k+1}+\T{E}_{k+1}-\T{X}\|_\infty\leq\epsilon$;
	\end{itemize}
	\textbf{end while} \\
	\label{alg1}	
\end{algorithm} 
As discussed in \cite{zhang2014novel},  the updates of $\T{L}_{k+1}$ and $\T{E}_{k+1}$ have closed form solutions.  It is easy to see that the main per-iteration cost of Algorithm \ref{alg1} is in the update of $\T{L}_{k+1}$, which requires computing the \texttt{fft} of $\T{L}$ and the SVD of block matrix $\hat{L}=\texttt{fft}(\T{L})$ in the Fourier domain.

The next result establishes the global convergence of the ADMM for solving problem \eqref{tlow7} (for details on the convergence analysis, see \cite{tarzanagh2018estimation}). Note that similar results hold for solving problems \eqref{tlow6}, and  \eqref{tlow8}. 
\begin{theorem}\label{admm:con}
The sequence $(\T{L}_k,\T{E}_k,\T{Y}_k)$ generated by Algorithm~\ref{alg1} from any starting point converges to a stationary point of problem \eqref{tlow7}. 
\end{theorem}

\section*{Acknowledgments}
The authors would like to thank the Associate Editor and three anonymous referees for many constructive comments and suggestions that improved significantly the structure and readability of the paper. 

\bibliographystyle{siamplain}
\bibliography{references}
\end{document}